\documentclass[aps,prl,reprint,amsmath,amssymb,superscriptaddress,
  longbibliography,nofootinbib]{revtex4-2}
\usepackage[T1]{fontenc}
\usepackage[utf8]{inputenc}
\usepackage{lmodern}
\usepackage{amsthm}
\usepackage{mathtools}
\usepackage{bm}
\usepackage{microtype}
\usepackage{graphicx}
\usepackage{tikz}
\usetikzlibrary{arrows.meta,positioning,calc}
\usepackage[colorlinks=true,linkcolor=blue,citecolor=blue,urlcolor=blue,
  pdfusetitle]{hyperref}

\newcommand{\R}{\mathbb R}
\newcommand{\E}{\mathbb E}
\newcommand{\Prob}{\mathbb P}
\newcommand{\id}{\mathbb I}
\newcommand{\ind}{\mathbf 1}
\newcommand{\sphere}{S^2}
\newcommand{\ket}[1]{\lvert #1\rangle}
\newcommand{\bra}[1]{\langle #1\rvert}
\newcommand{\norm}[1]{\lVert #1\rVert}
\newcommand{\abs}[1]{\lvert #1\rvert}
\newcommand{\pos}[1]{\bigl[#1\bigr]_+}
\newcommand{\dd}{\mathrm d}
\newcommand{\tv}{\mathrm{TV}}
\DeclareMathOperator{\diag}{diag}
\DeclareMathOperator{\sgn}{sgn}
\DeclareMathOperator{\clip}{clip}
\DeclareMathOperator{\Tr}{Tr}
\DeclareMathOperator{\Unif}{Unif}
\newcommand{\mainheading}[1]{%
  \par\medskip\noindent
  \textbf{#1.---}\enspace\ignorespaces
}
\newcommand{\protocolstep}[2]{%
  \par\smallskip\noindent
  \textit{Step #1: #2.}\enspace\ignorespaces
}

\newtheoremstyle{manuscript}{5pt}{5pt}{\itshape}{}{\bfseries}{.}{.5em}{}
\theoremstyle{manuscript}
\newtheorem{theorem}{Theorem}
\newtheorem{corollary}{Corollary}
\newtheorem{smtheorem}{Theorem}

\newtheorem{smlemma}{Lemma}

\newtheorem{smprop}{Proposition}

\begin{document}
\nocite{Bell1964,PR1994,TonerBacon2003,Cerf2005,Brunner2005,Brunner2008,Ahanj2011,Gisin2009,EPR1992,Renner2023,Portmann2012,Degorre2005,Chaturvedi2017,vanDam2013,Kaplan2011,Oszmaniec2017,Devroye1986,Strassen1965,Artstein1983}

\title{Finite Nonsignalling Resources Suffice to Simulate Entangled Qubits}

\author{Ziao Tang}
\author{Xin Wang}
\email{felixxinwang@hkust-gz.edu.cn}
\author{Ge Bai}
\email{gebai@hkust-gz.edu.cn}
\affiliation{The Hong Kong University of Science and Technology (Guangzhou),
Guangdong 511453, China}

\begin{abstract}
A longstanding question in Bell nonlocality is whether finite-alphabet nonsignalling resources can reproduce the continuum of measurement correlations of partially entangled qubits. Here we resolve this problem by showing that every fixed pure two-qubit state admits an exact simulation of all local projective-measurement statistics using shared randomness and finitely many Popescu-Rohrlich (PR) boxes, without communication or postselection. The boxes operate in parallel, with a worst-case bound independent of the measurement settings: at most $128\cot^4\theta$ calls for the Schmidt state $\cos\theta|00\rangle+\sin\theta|11\rangle$, $0<\theta<\pi/4$. The key is an exact finite-selection method that reduces input-dependent continuous sampling to a choice from a fixed-length shared random list, implemented through a finite PR-box network. The construction extends to arbitrary mixed two-qubit states and local binary-outcome POVMs. This reveals a fundamental separation between the continuity of quantum measurement correlations and the alphabet size of the resources needed to simulate them, shedding light on the resource-theoretic structure of Bell nonlocality.
\end{abstract}

\maketitle

\mainheading{\textit{Introduction}}
Bell's theorem shows that shared randomness alone cannot account for all quantum correlations~\cite{Bell1964}. It leaves a complementary question: how simple can the nonlocal resources be that reproduce these correlations without communication? Popescu--Rohrlich (PR) boxes provide a natural starting point, exhibiting correlations beyond the quantum set while obeying the nonsignalling principle, with only binary inputs and outputs~\cite{PR1994}. For qubits, the question is particularly intriguing: local projective measurements have only two outcomes, yet their measurement directions vary continuously.

For maximally entangled qubits, a single PR box and shared randomness suffice to simulate all local projective-measurement statistics~\cite{Cerf2005}, just as one bit of classical communication suffices in a different simulation model~\cite{TonerBacon2003}. Surprisingly, some partially entangled states cannot be simulated with one PR box~\cite{Brunner2005}, revealing that weaker entanglement need not entail simpler nonlocal resources. General exact nonsignalling protocols for partially entangled qubits have instead relied on an additional millionaire box (M-box), which compares real-valued inputs~\cite{Brunner2008,Ahanj2011}. Is this continuous input alphabet fundamentally necessary, or merely a feature of existing constructions? This finite-alphabet question, identified in Ref.~\cite{Brunner2008} and explicitly posed by Gisin~\cite{Gisin2009}, has remained open. A strong resolution requires exact simulation across all measurement directions with a finite worst-case resource budget, rather than discretizing the settings or bounding only the expected number of calls. Continuous shared randomness remains allowed, while the finiteness requirement concerns the nonsignalling resources.

Here we show that a continuous input alphabet is unnecessary. Every fixed pure two-qubit state admits an exact simulation of all local projective-measurement statistics using shared randomness and finitely many parallel PR boxes, without communication or postselection. For the Schmidt state $\cos\theta\ket{00}+\sin\theta\ket{11}$, $0<\theta<\pi/4$, the worst-case number of PR-box calls is at most $128\cot^4\theta$, independently of the measurement settings. The construction further extends to every fixed mixed two-qubit state and arbitrary local binary-outcome positive-operator-valued measurements (POVMs).

The key is to separate local sample selection from nonsignalling coordination. Starting from a known local--nonlocal decomposition~\cite{EPR1992,Renner2023} attaining the maximal local weight~\cite{Portmann2012}, we represent its nonlocal remainder through local responses to a random direction drawn from a distribution depending on Alice's measurement setting, connecting the problem to distributed sampling~\cite{Degorre2005}. Alice then selects from a fixed-length list of shared random directions prepared before the measurement settings arrive; averaging over the list and her local choice yields the required distribution exactly. Parallel PR boxes then coordinate the outputs to reproduce the corresponding correlations without revealing Alice's selected index to Bob (Fig.~\ref{fig:protocol}). Our result thus establishes that continuously varying qubit measurement settings do not require continuous-alphabet nonlocal resources for exact simulation, while providing a finite-selection approach applicable to other correlation families admitting suitable distributed-sampling representations.

\mainheading{\textit{Simulation task}}
Up to local changes of basis, any partially entangled pure
two-qubit state can be written as
\begin{equation}
 \ket{\psi_\theta}
 =\cos\theta\ket{00}+\sin\theta\ket{11},
 \qquad 0<\theta<\frac{\pi}{4}.
 \label{eq:state}
\end{equation}
Alice and Bob know the state and independently receive
measurement directions $x,y\in\sphere$, where $\sphere$
is the unit Bloch sphere. They must output
$A,B\in\{-1,+1\}$ according to the joint statistics of
the corresponding local projective measurements,
$\Pi_{\alpha|x}=(\id_2+\alpha x\cdot\sigma)/2$ and
$\Pi_{\beta|y}=(\id_2+\beta y\cdot\sigma)/2$,
where $\sigma$ denotes the Pauli vector.
Writing $c=\cos(2\theta)$, $s=\sin(2\theta)$, and
$T=\diag(s,-s,1)$,
the target quantum correlations are specified by
\begin{equation}
 P_\theta(\alpha,\beta\mid x,y)
 =\frac14\left[
 1+\alpha cx_z+\beta cy_z
 +\alpha\beta x^{\mathsf T}Ty
 \right],
 \label{eq:born}
\end{equation}
where $\alpha,\beta\in\{-1,+1\}$ and $x_z,y_z$ are
the third components of the measurement directions.
Thus, exact simulation requires reproducing both local
averages, $\E(A\mid x)=cx_z$ and $\E(B\mid y)=cy_z$,
and the correlator $\E(AB\mid x,y)=x^{\mathsf T}Ty$,
for every pair of measurement settings.

Alice and Bob may share randomness drawn independently of their
measurement inputs, with a distribution that may depend on the known state.
In a local model, each party generates an output using only this
shared randomness, their own measurement setting, and local randomness.
We additionally allow nonsignalling boxes, whose local
output marginals are independent of the distant input.
In particular, a PR box is defined by
\begin{equation}
 R_{\mathrm{PR}}(u,v\mid a,b)
 =\frac12\ind\{u\oplus v=ab\},
 \qquad a,b,u,v\in\{0,1\}.
 \label{eq:pr}
\end{equation}

The simulation must be exact for every $x,y$, with
no communication or postselection.
Every run must produce an output on both sides.
Crucially, the nonsignalling resources must have finite
input and output alphabets, and their number of calls
must be bounded in the worst case independently of
the measurement settings. Continuous shared randomness
and exact Borel-measurable local processing remain
allowed and no bound on their bit complexity is imposed.

\begin{theorem}[Finite nonsignalling simulation]
\label{thm:main}
For every fixed state $\ket{\psi_\theta}$ in
Eq.~\eqref{eq:state}, shared randomness and at most
\begin{equation}
 K(\theta)
 =4\left\lceil\sqrt{2}\pi\cot^2\theta\right\rceil^2
 \le128\cot^4\theta
 \label{eq:coarsebudget}
\end{equation}
PR-box calls suffice to reproduce all local
projective-measurement statistics exactly.
All calls can be made in parallel, without
communication or postselection.
\end{theorem}

The bound depends only on the state and scales as
$O(\theta^{-4})$ as $\theta\to0^+$.
The endpoints are covered separately:
the product state at $\theta=0$ requires no PR box,
whereas the maximally entangled state at
$\theta=\pi/4$ admits a one-PR-box simulation
\cite{Cerf2005}. Accordingly, we set
$K(0)=0$ and $K(\pi/4)=1$.

\mainheading{\textit{Local--nonlocal decomposition}}
To isolate the part of the quantum correlations that
requires nonlocal resources, we employ a known
local--nonlocal decomposition of $P_\theta$
\cite{EPR1992,Portmann2012,Renner2023}.
For pure two-qubit states, the maximal local weight is
$c=\cos(2\theta)$~\cite{Portmann2012}. An explicit
construction attaining this weight
\cite{Renner2023} gives
\begin{equation}
 P_\theta=cP_L+(1-c)P_N,
 \label{eq:decomposition}
\end{equation}
where $P_L$ is local and $P_N$ is nonsignalling.
A shared random coin selects the local branch with
probability $c$ and the nonlocal branch with probability
$1-c$, independently of the measurement settings.
Only the latter requires nonsignalling resources.

The local branch uses only shared randomness, following
Ref.~\cite{Renner2023} (reviewed in SM Sec.~\ref{sm:quantum}).

The remaining task is to simulate $P_N$. In the
decomposition above, $P_L$ has local means $x_z,y_z$,
whereas $P_N$ has uniform marginals and is determined
entirely by its correlator $C_N(x,y)$:
\begin{equation}
 P_N(\alpha,\beta\mid x,y)
 =\frac14\left[1+\alpha\beta C_N(x,y)\right].
 \label{eq:nonlocalbehavior}
\end{equation}
Crucially, this correlator admits a distributed-sampling
representation~\cite{Renner2023,Degorre2005}.
Let $\mu=\dd S/(4\pi)$ denote the uniform probability
measure on $\sphere$. For each measurement direction
$x$, define a probability distribution
$\nu_x(\dd\lambda)=f_x(\lambda)\mu(\dd\lambda)$,
with density
\begin{equation}
 f_x(\lambda)=\frac{2}{1-c}
 \pos{\abs{(Tx)\cdot\lambda}-c\abs{\lambda_z}},
 \label{eq:qubitdensity}
\end{equation}
where $[t]_+=\max(t,0)$. Define the local sign
responses
\begin{equation}
 a_x(\lambda)=\sgn((Tx)\cdot\lambda),\qquad
 b_y(\lambda)=\sgn(y\cdot\lambda),
 \label{eq:signresponses}
\end{equation}
with $\sgn(0)=+1$. The nonlocal correlator then reads
\begin{equation}
 C_N(x,y)=
 \int_{\sphere}a_x(\lambda)b_y(\lambda)
 \,\nu_x(\dd\lambda).
 \label{eq:correlationrepresentation}
\end{equation}
For fixed $\theta$, the density $f_x$ is normalized,
is even in $\lambda$, and is jointly continuous in
$(x,\lambda)$. These properties and the decomposition
are verified in SM Sec.~\ref{sm:quantum}.

Equation~\eqref{eq:correlationrepresentation} reveals
the central challenge: the sampled direction must
follow a distribution depending on Alice's input $x$,
yet the resulting correlations must be reproduced
without communicating that direction. In Protocol~6 of
Ref.~\cite{Renner2023}, Alice transmits the sampled
direction to Bob on the nonlocal branch. We overcome
this obstacle by replacing the transmitted direction
with an exact selection from a fixed-length list of
shared random candidates, whose selected response
is implemented using finitely many parallel PR boxes.
The resulting protocol reproduces the nonlocal
correlations without communicating the selected
direction or its index.

\begin{figure*}[t]
\centering
\begin{tikzpicture}[x=1cm,y=1cm,
  every node/.style={font=\footnotesize,align=center,inner sep=2pt},
  arr/.style={-{Latex[length=1.6mm,width=1.15mm]},line width=.6pt,
              shorten <=1pt,shorten >=1pt},
  card/.style={draw=black!35,rounded corners=2pt,line width=.45pt,fill=white},
  heading/.style={font=\small\bfseries,anchor=west}]
 \definecolor{regional}{RGB}{36,111,160}
 \definecolor{residual}{RGB}{202,100,36}
 \path[use as bounding box] (0,.05) rectangle (17.6,5.93);
 \draw[black!15] (4.4,.25) -- (4.4,5.35);
 \draw[black!15] (10.65,.25) -- (10.65,5.35);

 \node[heading] at (.05,5.68) {(a) Shared candidates};
 \definecolor{sphereLight}{RGB}{245,250,253}
 \definecolor{sphereMid}{RGB}{191,216,232}
 \definecolor{sphereEdge}{RGB}{111,154,184}
 \pgfdeclareradialshading{softSphere}{\pgfpoint{-12bp}{14bp}}{
   color(0bp)=(white); color(13bp)=(sphereLight);
   color(30bp)=(sphereMid); color(50bp)=(sphereEdge)}
 \begin{scope}[shift={(2.15,3.92)}]
  \shade[shading=softSphere] (0,0) circle (1.43);
  \begin{scope}
   \clip (0,0) circle (1.43);
   \fill[regional,opacity=.10] (0.22831,-0.71385) -- (0.26518,-0.72296) -- (0.30184,-0.73150) -- (0.33826,-0.73947) -- (0.37441,-0.74685) -- (0.41027,-0.75364) -- (0.44581,-0.75985) -- (0.48100,-0.76545) -- (0.51581,-0.77045) -- (0.55021,-0.77485) -- (0.58418,-0.77864) -- (0.61769,-0.78181) -- (0.65071,-0.78437) -- (0.68322,-0.78631) -- (0.71520,-0.78763) -- (0.74661,-0.78833) -- (0.77743,-0.78842) -- (0.80764,-0.78788) -- (0.83722,-0.78672) -- (0.86614,-0.78495) -- (0.89438,-0.78255) -- (0.92191,-0.77954) -- (0.94872,-0.77592) -- (0.97478,-0.77169) -- (1.00008,-0.76685) -- (1.02458,-0.76140) -- (1.04829,-0.75536) -- (1.07116,-0.74872) -- (1.09320,-0.74150) -- (1.09320,-0.74150) -- (1.10564,-0.71113) -- (1.11738,-0.68030) -- (1.12842,-0.64898) -- (1.13876,-0.61720) -- (1.14839,-0.58493) -- (1.15731,-0.55217) -- (1.16550,-0.51893) -- (1.17296,-0.48519) -- (1.17968,-0.45094) -- (1.18565,-0.41619) -- (1.19084,-0.38091) -- (1.19525,-0.34510) -- (1.19885,-0.30874) -- (1.20163,-0.27182) -- (1.20355,-0.23432) -- (1.20460,-0.19623) -- (1.20473,-0.15752) -- (1.20392,-0.11816) -- (1.20211,-0.07814) -- (1.19928,-0.03742) -- (1.19536,0.00404) -- (1.19030,0.04628) -- (1.18402,0.08934) -- (1.17645,0.13328) -- (1.16749,0.17816) -- (1.15703,0.22405) -- (1.14495,0.27105) -- (1.13108,0.31925) -- (1.13108,0.31925) -- (1.11550,0.31415) -- (1.09933,0.30945) -- (1.08257,0.30518) -- (1.06524,0.30133) -- (1.04735,0.29791) -- (1.02892,0.29491) -- (1.00997,0.29235) -- (0.99050,0.29022) -- (0.97053,0.28853) -- (0.95008,0.28728) -- (0.92917,0.28646) -- (0.90780,0.28608) -- (0.88601,0.28614) -- (0.86380,0.28663) -- (0.84119,0.28757) -- (0.81820,0.28894) -- (0.79485,0.29075) -- (0.77115,0.29300) -- (0.74713,0.29567) -- (0.72280,0.29878) -- (0.69819,0.30232) -- (0.67331,0.30628) -- (0.64818,0.31067) -- (0.62283,0.31547) -- (0.59726,0.32069) -- (0.57151,0.32632) -- (0.54559,0.33236) -- (0.51952,0.33880) -- (0.51952,0.33880) -- (0.51230,0.29127) -- (0.50470,0.24490) -- (0.49677,0.19959) -- (0.48852,0.15527) -- (0.47999,0.11184) -- (0.47118,0.06927) -- (0.46212,0.02748) -- (0.45282,-0.01356) -- (0.44329,-0.05389) -- (0.43355,-0.09354) -- (0.42361,-0.13255) -- (0.41346,-0.17094) -- (0.40312,-0.20874) -- (0.39261,-0.24596) -- (0.38191,-0.28263) -- (0.37104,-0.31876) -- (0.36000,-0.35435) -- (0.34880,-0.38944) -- (0.33743,-0.42402) -- (0.32591,-0.45811) -- (0.31424,-0.49172) -- (0.30241,-0.52485) -- (0.29043,-0.55750) -- (0.27830,-0.58969) -- (0.26602,-0.62142) -- (0.25360,-0.65269) -- (0.24103,-0.68350) -- (0.22831,-0.71385) -- cycle;
   \draw[regional!45,densely dashed,line width=.45pt,opacity=.45] (-0.22831,0.71385) -- (-0.26518,0.72296) -- (-0.30184,0.73150) -- (-0.33826,0.73947) -- (-0.37441,0.74685) -- (-0.41027,0.75364) -- (-0.44581,0.75985) -- (-0.48100,0.76545) -- (-0.51581,0.77045) -- (-0.55021,0.77485) -- (-0.58418,0.77864) -- (-0.61769,0.78181) -- (-0.65071,0.78437) -- (-0.68322,0.78631) -- (-0.71520,0.78763) -- (-0.74661,0.78833) -- (-0.77743,0.78842) -- (-0.80764,0.78788) -- (-0.83722,0.78672) -- (-0.86614,0.78495) -- (-0.89438,0.78255) -- (-0.92191,0.77954) -- (-0.94872,0.77592) -- (-0.97478,0.77169) -- (-1.00008,0.76685) -- (-1.02458,0.76140) -- (-1.04829,0.75536) -- (-1.07116,0.74872) -- (-1.09320,0.74150) -- (-1.09320,0.74150) -- (-1.10564,0.71113) -- (-1.11738,0.68030) -- (-1.12842,0.64898) -- (-1.13876,0.61720) -- (-1.14839,0.58493) -- (-1.15731,0.55217) -- (-1.16550,0.51893) -- (-1.17296,0.48519) -- (-1.17968,0.45094) -- (-1.18565,0.41619) -- (-1.19084,0.38091) -- (-1.19525,0.34510) -- (-1.19885,0.30874) -- (-1.20163,0.27182) -- (-1.20355,0.23432) -- (-1.20460,0.19623) -- (-1.20473,0.15752) -- (-1.20392,0.11816) -- (-1.20211,0.07814) -- (-1.19928,0.03742) -- (-1.19536,-0.00404) -- (-1.19030,-0.04628) -- (-1.18402,-0.08934) -- (-1.17645,-0.13328) -- (-1.16749,-0.17816) -- (-1.15703,-0.22405) -- (-1.14495,-0.27105) -- (-1.13108,-0.31925) -- (-1.13108,-0.31925) -- (-1.11550,-0.31415) -- (-1.09933,-0.30945) -- (-1.08257,-0.30518) -- (-1.06524,-0.30133) -- (-1.04735,-0.29791) -- (-1.02892,-0.29491) -- (-1.00997,-0.29235) -- (-0.99050,-0.29022) -- (-0.97053,-0.28853) -- (-0.95008,-0.28728) -- (-0.92917,-0.28646) -- (-0.90780,-0.28608) -- (-0.88601,-0.28614) -- (-0.86380,-0.28663) -- (-0.84119,-0.28757) -- (-0.81820,-0.28894) -- (-0.79485,-0.29075) -- (-0.77115,-0.29300) -- (-0.74713,-0.29567) -- (-0.72280,-0.29878) -- (-0.69819,-0.30232) -- (-0.67331,-0.30628) -- (-0.64818,-0.31067) -- (-0.62283,-0.31547) -- (-0.59726,-0.32069) -- (-0.57151,-0.32632) -- (-0.54559,-0.33236) -- (-0.51952,-0.33880) -- (-0.51952,-0.33880) -- (-0.51230,-0.29127) -- (-0.50470,-0.24490) -- (-0.49677,-0.19959) -- (-0.48852,-0.15527) -- (-0.47999,-0.11184) -- (-0.47118,-0.06927) -- (-0.46212,-0.02748) -- (-0.45282,0.01356) -- (-0.44329,0.05389) -- (-0.43355,0.09354) -- (-0.42361,0.13255) -- (-0.41346,0.17094) -- (-0.40312,0.20874) -- (-0.39261,0.24596) -- (-0.38191,0.28263) -- (-0.37104,0.31876) -- (-0.36000,0.35435) -- (-0.34880,0.38944) -- (-0.33743,0.42402) -- (-0.32591,0.45811) -- (-0.31424,0.49172) -- (-0.30241,0.52485) -- (-0.29043,0.55750) -- (-0.27830,0.58969) -- (-0.26602,0.62142) -- (-0.25360,0.65269) -- (-0.24103,0.68350) -- (-0.22831,0.71385) -- cycle;
   \draw[draw=regional!30,line width=.3pt,densely dashed,opacity=.38] (-1.13108,-0.31925) -- (-1.14962,-0.32619) -- (-1.16720,-0.33375) -- (-1.18378,-0.34194) -- (-1.19937,-0.35074) -- (-1.21392,-0.36013) -- (-1.22744,-0.37012) -- (-1.23989,-0.38068) -- (-1.25128,-0.39181) -- (-1.26157,-0.40349) -- (-1.27076,-0.41570) -- (-1.27884,-0.42843) -- (-1.28580,-0.44168) -- (-1.29163,-0.45541) -- (-1.29632,-0.46961) -- (-1.29987,-0.48427) -- (-1.30227,-0.49936) -- (-1.30352,-0.51488) -- (-1.30361,-0.53080) -- (-1.30256,-0.54709) -- (-1.30036,-0.56375) -- (-1.29700,-0.58075) -- (-1.29251,-0.59807) -- (-1.28687,-0.61569) -- (-1.28011,-0.63359) -- (-1.27221,-0.65174) -- (-1.26321,-0.67013) -- (-1.25926,-0.67761);
   \draw[draw=regional!30,line width=.3pt,densely dashed,opacity=.38] (0.38733,-1.37655) -- (0.39545,-1.37419) -- (0.41493,-1.36789) -- (0.43347,-1.36096) -- (0.45104,-1.35339) -- (0.46763,-1.34521) -- (0.48322,-1.33641) -- (0.49777,-1.32701) -- (0.51129,-1.31703) -- (0.52374,-1.30646) -- (0.53512,-1.29534) -- (0.54542,-1.28366) -- (0.55461,-1.27145) -- (0.56269,-1.25871) -- (0.56965,-1.24547) -- (0.57548,-1.23174) -- (0.58017,-1.21754) -- (0.58372,-1.20288) -- (0.58612,-1.18778) -- (0.58737,-1.17227) -- (0.58746,-1.15635) -- (0.58641,-1.14005) -- (0.58420,-1.12339) -- (0.58085,-1.10639) -- (0.57636,-1.08907) -- (0.57072,-1.07145) -- (0.56395,-1.05356) -- (0.55606,-1.03541) -- (0.54706,-1.01702) -- (0.53695,-0.99842) -- (0.52576,-0.97964) -- (0.51348,-0.96068) -- (0.50015,-0.94159) -- (0.48577,-0.92238) -- (0.47036,-0.90307) -- (0.45394,-0.88368) -- (0.43653,-0.86425) -- (0.41815,-0.84480) -- (0.39883,-0.82534) -- (0.37859,-0.80590) -- (0.35745,-0.78651) -- (0.33543,-0.76719) -- (0.31257,-0.74796) -- (0.28890,-0.72885) -- (0.26444,-0.70988) -- (0.23921,-0.69107) -- (0.21326,-0.67245) -- (0.18662,-0.65404) -- (0.15931,-0.63585) -- (0.13137,-0.61792) -- (0.10283,-0.60027) -- (0.07374,-0.58291) -- (0.04411,-0.56587) -- (0.01400,-0.54917) -- (-0.01657,-0.53282) -- (-0.04755,-0.51686) -- (-0.07891,-0.50129) -- (-0.11061,-0.48614) -- (-0.14262,-0.47143) -- (-0.17488,-0.45716) -- (-0.20737,-0.44337) -- (-0.24004,-0.43007) -- (-0.27286,-0.41727) -- (-0.30578,-0.40499) -- (-0.33876,-0.39325) -- (-0.37177,-0.38205) -- (-0.40476,-0.37141) -- (-0.43770,-0.36135) -- (-0.47053,-0.35188) -- (-0.50323,-0.34301) -- (-0.53575,-0.33475) -- (-0.56806,-0.32710) -- (-0.60011,-0.32009) -- (-0.63187,-0.31371) -- (-0.66329,-0.30798) -- (-0.69434,-0.30291) -- (-0.72498,-0.29849) -- (-0.75517,-0.29473) -- (-0.78488,-0.29165) -- (-0.81407,-0.28923) -- (-0.84271,-0.28749) -- (-0.87075,-0.28643) -- (-0.89817,-0.28605) -- (-0.92493,-0.28635) -- (-0.95100,-0.28732) -- (-0.97635,-0.28898) -- (-1.00094,-0.29131) -- (-1.02475,-0.29431) -- (-1.04775,-0.29798) -- (-1.06991,-0.30231) -- (-1.09120,-0.30731) -- (-1.11160,-0.31296) -- (-1.13108,-0.31925);
   \draw[draw=regional!30,line width=.3pt,densely dashed,opacity=.38] (-1.09320,0.74150) -- (-1.11942,0.73169) -- (-1.14427,0.72099) -- (-1.16773,0.70942) -- (-1.18977,0.69698) -- (-1.21035,0.68369) -- (-1.22947,0.66957) -- (-1.24708,0.65463) -- (-1.26318,0.63889) -- (-1.27773,0.62238) -- (-1.29073,0.60510) -- (-1.30216,0.58709) -- (-1.31200,0.56837) -- (-1.31632,0.55875);
   \draw[draw=regional!30,line width=.3pt,densely dashed,opacity=.38] (1.31632,-0.55875) -- (1.32024,-0.54895) -- (1.32688,-0.52887) -- (1.33189,-0.50813) -- (1.33529,-0.48678) -- (1.33706,-0.46484) -- (1.33719,-0.44233) -- (1.33570,-0.41928) -- (1.33259,-0.39573) -- (1.32784,-0.37168) -- (1.32149,-0.34719) -- (1.31352,-0.32227) -- (1.30395,-0.29696) -- (1.29279,-0.27129) -- (1.28005,-0.24529) -- (1.26576,-0.21899) -- (1.24993,-0.19242) -- (1.23257,-0.16562) -- (1.21371,-0.13862) -- (1.19337,-0.11144) -- (1.17158,-0.08413) -- (1.14836,-0.05672) -- (1.12374,-0.02924) -- (1.09775,-0.00173) -- (1.07043,0.02579) -- (1.04180,0.05328) -- (1.01190,0.08070) -- (0.98077,0.10802) -- (0.94844,0.13521) -- (0.91496,0.16224) -- (0.88036,0.18907) -- (0.84469,0.21567) -- (0.80800,0.24200) -- (0.77031,0.26804) -- (0.73169,0.29376) -- (0.69218,0.31912) -- (0.65182,0.34408) -- (0.61067,0.36863) -- (0.56878,0.39273) -- (0.52619,0.41635) -- (0.48296,0.43947) -- (0.43915,0.46205) -- (0.39480,0.48406) -- (0.34996,0.50549) -- (0.30470,0.52630) -- (0.25907,0.54647) -- (0.21313,0.56597) -- (0.16692,0.58478) -- (0.12051,0.60288) -- (0.07396,0.62025) -- (0.02731,0.63686) -- (-0.01937,0.65269) -- (-0.06602,0.66773) -- (-0.11260,0.68196) -- (-0.15904,0.69536) -- (-0.20528,0.70790) -- (-0.25128,0.71959) -- (-0.29696,0.73040) -- (-0.34229,0.74032) -- (-0.38720,0.74933) -- (-0.43164,0.75744) -- (-0.47555,0.76462) -- (-0.51888,0.77087) -- (-0.56158,0.77618) -- (-0.60360,0.78054) -- (-0.64488,0.78396) -- (-0.68537,0.78642) -- (-0.72503,0.78792) -- (-0.76381,0.78846) -- (-0.80165,0.78804) -- (-0.83852,0.78666) -- (-0.87437,0.78432) -- (-0.90915,0.78102) -- (-0.94283,0.77678) -- (-0.97535,0.77159) -- (-1.00669,0.76545) -- (-1.03680,0.75839) -- (-1.06565,0.75040) -- (-1.09320,0.74150);
   \draw[draw=regional!30,line width=.3pt,densely dashed,opacity=.38] (1.25926,0.67761) -- (1.25310,0.68872) -- (1.24191,0.70751) -- (1.22963,0.72646) -- (1.21630,0.74556) -- (1.20192,0.76477) -- (1.18651,0.78408) -- (1.17009,0.80346) -- (1.15268,0.82289) -- (1.13430,0.84235) -- (1.11498,0.86181) -- (1.09474,0.88125) -- (1.07360,0.90064) -- (1.05158,0.91996) -- (1.02873,0.93918) -- (1.00505,0.95829) -- (0.98059,0.97726) -- (0.95537,0.99607) -- (0.92942,1.01470) -- (0.90277,1.03311) -- (0.87546,1.05129) -- (0.84752,1.06922) -- (0.81898,1.08688) -- (0.78989,1.10424) -- (0.76026,1.12128) -- (0.73015,1.13798) -- (0.69958,1.15432) -- (0.66860,1.17029) -- (0.63724,1.18586) -- (0.60554,1.20101) -- (0.57353,1.21572) -- (0.54127,1.22998) -- (0.50878,1.24377) -- (0.47611,1.25708) -- (0.44329,1.26988) -- (0.41037,1.28216) -- (0.37739,1.29390) -- (0.34438,1.30510) -- (0.31139,1.31573) -- (0.27846,1.32579) -- (0.24562,1.33526) -- (0.21292,1.34414) -- (0.18040,1.35240) -- (0.14809,1.36004) -- (0.11604,1.36706) -- (0.08428,1.37343) -- (0.05286,1.37916) -- (0.02181,1.38424) -- (-0.00883,1.38866) -- (-0.03902,1.39241) -- (-0.06873,1.39550) -- (-0.09792,1.39792) -- (-0.12655,1.39965) -- (-0.15460,1.40071) -- (-0.18202,1.40110) -- (-0.20878,1.40080) -- (-0.23485,1.39982) -- (-0.26020,1.39817) -- (-0.28479,1.39584) -- (-0.30860,1.39284) -- (-0.33160,1.38917) -- (-0.35376,1.38483) -- (-0.37505,1.37984) -- (-0.38733,1.37655);
   \draw[draw=regional!30,line width=.3pt,densely dashed,opacity=.38] (-0.50640,-1.19299) -- (-0.54424,-1.16639) -- (-0.58142,-1.13836) -- (-0.61789,-1.10895) -- (-0.65361,-1.07819) -- (-0.68853,-1.04611) -- (-0.72262,-1.01276) -- (-0.75582,-0.97817) -- (-0.78810,-0.94239) -- (-0.81943,-0.90547) -- (-0.84975,-0.86744) -- (-0.87904,-0.82835) -- (-0.90726,-0.78826) -- (-0.93437,-0.74720) -- (-0.96035,-0.70524) -- (-0.98515,-0.66241) -- (-1.00875,-0.61878) -- (-1.03113,-0.57440) -- (-1.05225,-0.52931) -- (-1.07208,-0.48358) -- (-1.09062,-0.43726) -- (-1.10782,-0.39041) -- (-1.12367,-0.34308) -- (-1.13815,-0.29533) -- (-1.15125,-0.24723) -- (-1.16294,-0.19882) -- (-1.17322,-0.15017) -- (-1.18207,-0.10134) -- (-1.18947,-0.05238) -- (-1.19543,-0.00336) -- (-1.19993,0.04566) -- (-1.20297,0.09463) -- (-1.20455,0.14348) -- (-1.20466,0.19216) -- (-1.20329,0.24060) -- (-1.20047,0.28875) -- (-1.19618,0.33655) -- (-1.19043,0.38394) -- (-1.18323,0.43086) -- (-1.17459,0.47726) -- (-1.16452,0.52307) -- (-1.15304,0.56825) -- (-1.14014,0.61273) -- (-1.12586,0.65647) -- (-1.11020,0.69941) -- (-1.09320,0.74150) -- (-1.07486,0.78268) -- (-1.05521,0.82291) -- (-1.03428,0.86214) -- (-1.01208,0.90031) -- (-0.98865,0.93739) -- (-0.96402,0.97333) -- (-0.93822,1.00808) -- (-0.91127,1.04161) -- (-0.88321,1.07386) -- (-0.85407,1.10481) -- (-0.82390,1.13441) -- (-0.79272,1.16263) -- (-0.76057,1.18943) -- (-0.72750,1.21478) -- (-0.69354,1.23865) -- (-0.65874,1.26102) -- (-0.62313,1.28184) -- (-0.58676,1.30111) -- (-0.54968,1.31879) -- (-0.51193,1.33486) -- (-0.47356,1.34931) -- (-0.46996,1.35057);
   \draw[draw=regional!30,line width=.3pt,densely dashed,opacity=.38] (0.46996,-1.35057) -- (0.43461,-1.36211) -- (0.39513,-1.37326) -- (0.35517,-1.38273) -- (0.31477,-1.39051) -- (0.27399,-1.39660) -- (0.23288,-1.40099) -- (0.19149,-1.40368) -- (0.14986,-1.40465) -- (0.10805,-1.40391) -- (0.06610,-1.40146) -- (0.02408,-1.39731) -- (-0.01797,-1.39145) -- (-0.06000,-1.38389) -- (-0.10196,-1.37465) -- (-0.14379,-1.36374) -- (-0.18545,-1.35116) -- (-0.22688,-1.33694) -- (-0.26804,-1.32109) -- (-0.30887,-1.30363) -- (-0.34932,-1.28458) -- (-0.38935,-1.26397) -- (-0.42890,-1.24181) -- (-0.46793,-1.21814) -- (-0.50640,-1.19299);
   \draw[draw=regional!30,line width=.3pt,densely dashed,opacity=.38] (-0.50640,-1.19299) -- (-0.55207,-1.18058) -- (-0.59708,-1.16673) -- (-0.64136,-1.15146) -- (-0.68486,-1.13479) -- (-0.72752,-1.11673) -- (-0.76929,-1.09732) -- (-0.81013,-1.07656) -- (-0.84999,-1.05450) -- (-0.88880,-1.03115) -- (-0.92654,-1.00654) -- (-0.96314,-0.98071) -- (-0.99857,-0.95368) -- (-1.03279,-0.92550) -- (-1.06574,-0.89618) -- (-1.09740,-0.86577) -- (-1.12772,-0.83431) -- (-1.15667,-0.80183) -- (-1.18421,-0.76837) -- (-1.21030,-0.73398) -- (-1.23492,-0.69869) -- (-1.25804,-0.66255) -- (-1.27962,-0.62561) -- (-1.29965,-0.58790) -- (-1.31809,-0.54948) -- (-1.33492,-0.51038) -- (-1.35013,-0.47067) -- (-1.36203,-0.43562);
   \draw[draw=regional!30,line width=.3pt,densely dashed,opacity=.38] (1.36203,0.43562) -- (1.36369,0.43038) -- (1.37560,0.38957) -- (1.38582,0.34828) -- (1.39436,0.30657) -- (1.40120,0.26448) -- (1.40633,0.22207) -- (1.40975,0.17939) -- (1.41145,0.13650) -- (1.41143,0.09343) -- (1.40969,0.05026) -- (1.40624,0.00702) -- (1.40107,-0.03623) -- (1.39419,-0.07943) -- (1.38562,-0.12254) -- (1.37535,-0.16549) -- (1.36342,-0.20825) -- (1.34982,-0.25075) -- (1.33457,-0.29294) -- (1.31770,-0.33478) -- (1.29923,-0.37621) -- (1.27917,-0.41719) -- (1.25755,-0.45765) -- (1.23440,-0.49756) -- (1.20975,-0.53686) -- (1.18362,-0.57550) -- (1.15605,-0.61345) -- (1.12707,-0.65065) -- (1.09672,-0.68705) -- (1.06504,-0.72262) -- (1.03205,-0.75731) -- (0.99781,-0.79107) -- (0.96235,-0.82387) -- (0.92572,-0.85567) -- (0.88797,-0.88643) -- (0.84913,-0.91610) -- (0.80925,-0.94466) -- (0.76839,-0.97207) -- (0.72659,-0.99829) -- (0.68391,-1.02330) -- (0.64040,-1.04706) -- (0.59610,-1.06954) -- (0.55108,-1.09073) -- (0.50539,-1.11058) -- (0.45908,-1.12908) -- (0.41221,-1.14620) -- (0.36484,-1.16193) -- (0.31703,-1.17624) -- (0.26883,-1.18912) -- (0.22030,-1.20055) -- (0.17150,-1.21052) -- (0.12250,-1.21901) -- (0.07334,-1.22602) -- (0.02410,-1.23154) -- (-0.02518,-1.23555) -- (-0.07442,-1.23806) -- (-0.12357,-1.23906) -- (-0.17257,-1.23855) -- (-0.22136,-1.23653) -- (-0.26989,-1.23300) -- (-0.31808,-1.22798) -- (-0.36588,-1.22145) -- (-0.41324,-1.21344) -- (-0.46010,-1.20395) -- (-0.50640,-1.19299);
   \draw[draw=regional!30,line width=.3pt,densely dashed,opacity=.38] (-0.50640,-1.19299) -- (-0.53297,-1.20162) -- (-0.55890,-1.20878) -- (-0.58414,-1.21448) -- (-0.60867,-1.21869) -- (-0.63247,-1.22141) -- (-0.65549,-1.22265) -- (-0.67771,-1.22240) -- (-0.69911,-1.22066) -- (-0.71965,-1.21743) -- (-0.73932,-1.21272) -- (-0.75809,-1.20653) -- (-0.77593,-1.19888) -- (-0.79283,-1.18976) -- (-0.80271,-1.18345);
   \draw[draw=regional!30,line width=.3pt,densely dashed,opacity=.38] (0.80271,1.18345) -- (0.80876,1.17919) -- (0.82371,1.16718) -- (0.83765,1.15376) -- (0.85057,1.13892) -- (0.86246,1.12270) -- (0.87330,1.10512) -- (0.88307,1.08618) -- (0.89177,1.06592) -- (0.89938,1.04437) -- (0.90589,1.02154) -- (0.91130,0.99746) -- (0.91560,0.97218) -- (0.91878,0.94570) -- (0.92085,0.91808) -- (0.92179,0.88933) -- (0.92161,0.85950) -- (0.92031,0.82863) -- (0.91788,0.79675) -- (0.91434,0.76389) -- (0.90969,0.73010) -- (0.90392,0.69543) -- (0.89706,0.65991) -- (0.88910,0.62358) -- (0.88005,0.58649) -- (0.86994,0.54869) -- (0.85877,0.51022) -- (0.84655,0.47113) -- (0.83329,0.43147) -- (0.81903,0.39128) -- (0.80376,0.35061) -- (0.78752,0.30952) -- (0.77031,0.26804) -- (0.75217,0.22625) -- (0.73311,0.18417) -- (0.71316,0.14187) -- (0.69234,0.09940) -- (0.67068,0.05681) -- (0.64819,0.01415) -- (0.62492,-0.02853) -- (0.60089,-0.07117) -- (0.57613,-0.11373) -- (0.55066,-0.15615) -- (0.52452,-0.19838) -- (0.49775,-0.24036) -- (0.47036,-0.28206) -- (0.44241,-0.32341) -- (0.41391,-0.36436) -- (0.38491,-0.40488) -- (0.35545,-0.44489) -- (0.32555,-0.48437) -- (0.29525,-0.52326) -- (0.26459,-0.56151) -- (0.23361,-0.59907) -- (0.20235,-0.63591) -- (0.17083,-0.67197) -- (0.13912,-0.70721) -- (0.10723,-0.74159) -- (0.07521,-0.77507) -- (0.04310,-0.80760) -- (0.01093,-0.83915) -- (-0.02124,-0.86967) -- (-0.05339,-0.89914) -- (-0.08548,-0.92751) -- (-0.11746,-0.95475) -- (-0.14930,-0.98083) -- (-0.18096,-1.00571) -- (-0.21239,-1.02937) -- (-0.24357,-1.05177) -- (-0.27445,-1.07290) -- (-0.30500,-1.09271) -- (-0.33517,-1.11119) -- (-0.36494,-1.12832) -- (-0.39426,-1.14408) -- (-0.42310,-1.15844) -- (-0.45143,-1.17139) -- (-0.47920,-1.18291) -- (-0.50640,-1.19299);
   \draw[draw=regional!30,line width=.3pt,densely dashed,opacity=.38] (-0.50640,-1.19299) -- (-0.49812,-1.21718) -- (-0.48924,-1.23988) -- (-0.47976,-1.26108) -- (-0.46969,-1.28073) -- (-0.45906,-1.29883) -- (-0.44786,-1.31534) -- (-0.43612,-1.33025) -- (-0.42385,-1.34354) -- (-0.41106,-1.35520) -- (-0.39777,-1.36520) -- (-0.38399,-1.37354) -- (-0.36975,-1.38020) -- (-0.35506,-1.38519) -- (-0.35213,-1.38597);
   \draw[draw=regional!30,line width=.3pt,densely dashed,opacity=.38] (0.35213,1.38597) -- (0.33993,1.38848) -- (0.32439,1.39009) -- (0.30846,1.39000) -- (0.29215,1.38822) -- (0.27548,1.38474) -- (0.25848,1.37958) -- (0.24116,1.37274) -- (0.22355,1.36423) -- (0.20567,1.35405) -- (0.18754,1.34223) -- (0.16917,1.32877) -- (0.15061,1.31369) -- (0.13185,1.29700) -- (0.11294,1.27874) -- (0.09389,1.25893) -- (0.07473,1.23757) -- (0.05547,1.21471) -- (0.03615,1.19037) -- (0.01678,1.16458) -- (-0.00261,1.13737) -- (-0.02199,1.10878) -- (-0.04135,1.07883) -- (-0.06065,1.04757) -- (-0.07989,1.01503) -- (-0.09902,0.98126) -- (-0.11804,0.94629) -- (-0.13691,0.91017) -- (-0.15561,0.87294) -- (-0.17413,0.83465) -- (-0.19243,0.79533) -- (-0.21050,0.75505) -- (-0.22831,0.71385) -- (-0.24585,0.67179) -- (-0.26308,0.62890) -- (-0.28000,0.58524) -- (-0.29657,0.54088) -- (-0.31278,0.49585) -- (-0.32861,0.45022) -- (-0.34404,0.40404) -- (-0.35905,0.35737) -- (-0.37362,0.31026) -- (-0.38774,0.26278) -- (-0.40139,0.21497) -- (-0.41454,0.16690) -- (-0.42720,0.11863) -- (-0.43933,0.07022) -- (-0.45092,0.02172) -- (-0.46197,-0.02681) -- (-0.47245,-0.07530) -- (-0.48236,-0.12370) -- (-0.49168,-0.17195) -- (-0.50040,-0.22000) -- (-0.50851,-0.26777) -- (-0.51601,-0.31522) -- (-0.52287,-0.36228) -- (-0.52910,-0.40890) -- (-0.53468,-0.45503) -- (-0.53961,-0.50060) -- (-0.54388,-0.54556) -- (-0.54749,-0.58985) -- (-0.55044,-0.63343) -- (-0.55271,-0.67623) -- (-0.55431,-0.71822) -- (-0.55523,-0.75932) -- (-0.55548,-0.79950) -- (-0.55505,-0.83871) -- (-0.55394,-0.87689) -- (-0.55216,-0.91401) -- (-0.54971,-0.95001) -- (-0.54659,-0.98486) -- (-0.54280,-1.01850) -- (-0.53835,-1.05091) -- (-0.53324,-1.08203) -- (-0.52749,-1.11184) -- (-0.52109,-1.14029) -- (-0.51405,-1.16735) -- (-0.50640,-1.19299);
   \draw[draw=black!43,line width=.4pt,opacity=.65] (-1.25926,-0.67761) -- (-1.25310,-0.68872) -- (-1.24191,-0.70751) -- (-1.22963,-0.72646) -- (-1.21630,-0.74556) -- (-1.20192,-0.76477) -- (-1.18651,-0.78408) -- (-1.17009,-0.80346) -- (-1.15268,-0.82289) -- (-1.13430,-0.84235) -- (-1.11498,-0.86181) -- (-1.09474,-0.88125) -- (-1.07360,-0.90064) -- (-1.05158,-0.91996) -- (-1.02873,-0.93918) -- (-1.00505,-0.95829) -- (-0.98059,-0.97726) -- (-0.95537,-0.99607) -- (-0.92942,-1.01470) -- (-0.90277,-1.03311) -- (-0.87546,-1.05129) -- (-0.84752,-1.06922) -- (-0.81898,-1.08688) -- (-0.78989,-1.10424) -- (-0.76026,-1.12128) -- (-0.73015,-1.13798) -- (-0.69958,-1.15432) -- (-0.66860,-1.17029) -- (-0.63724,-1.18586) -- (-0.60554,-1.20101) -- (-0.57353,-1.21572) -- (-0.54127,-1.22998) -- (-0.50878,-1.24377) -- (-0.47611,-1.25708) -- (-0.44329,-1.26988) -- (-0.41037,-1.28216) -- (-0.37739,-1.29390) -- (-0.34438,-1.30510) -- (-0.31139,-1.31573) -- (-0.27846,-1.32579) -- (-0.24562,-1.33526) -- (-0.21292,-1.34414) -- (-0.18040,-1.35240) -- (-0.14809,-1.36004) -- (-0.11604,-1.36706) -- (-0.08428,-1.37343) -- (-0.05286,-1.37916) -- (-0.02181,-1.38424) -- (0.00883,-1.38866) -- (0.03902,-1.39241) -- (0.06873,-1.39550) -- (0.09792,-1.39792) -- (0.12655,-1.39965) -- (0.15460,-1.40071) -- (0.18202,-1.40110) -- (0.20878,-1.40080) -- (0.23485,-1.39982) -- (0.26020,-1.39817) -- (0.28479,-1.39584) -- (0.30860,-1.39284) -- (0.33160,-1.38917) -- (0.35376,-1.38483) -- (0.37505,-1.37984) -- (0.38733,-1.37655);
   \draw[draw=black!43,line width=.4pt,opacity=.65] (-1.31632,0.55875) -- (-1.32024,0.54895) -- (-1.32688,0.52887) -- (-1.33189,0.50813) -- (-1.33529,0.48678) -- (-1.33706,0.46484) -- (-1.33719,0.44233) -- (-1.33570,0.41928) -- (-1.33259,0.39573) -- (-1.32784,0.37168) -- (-1.32149,0.34719) -- (-1.31352,0.32227) -- (-1.30395,0.29696) -- (-1.29279,0.27129) -- (-1.28005,0.24529) -- (-1.26576,0.21899) -- (-1.24993,0.19242) -- (-1.23257,0.16562) -- (-1.21371,0.13862) -- (-1.19337,0.11144) -- (-1.17158,0.08413) -- (-1.14836,0.05672) -- (-1.12374,0.02924) -- (-1.09775,0.00173) -- (-1.07043,-0.02579) -- (-1.04180,-0.05328) -- (-1.01190,-0.08070) -- (-0.98077,-0.10802) -- (-0.94844,-0.13521) -- (-0.91496,-0.16224) -- (-0.88036,-0.18907) -- (-0.84469,-0.21567) -- (-0.80800,-0.24200) -- (-0.77031,-0.26804) -- (-0.73169,-0.29376) -- (-0.69218,-0.31912) -- (-0.65182,-0.34408) -- (-0.61067,-0.36863) -- (-0.56878,-0.39273) -- (-0.52619,-0.41635) -- (-0.48296,-0.43947) -- (-0.43915,-0.46205) -- (-0.39480,-0.48406) -- (-0.34996,-0.50549) -- (-0.30470,-0.52630) -- (-0.25907,-0.54647) -- (-0.21313,-0.56597) -- (-0.16692,-0.58478) -- (-0.12051,-0.60288) -- (-0.07396,-0.62025) -- (-0.02731,-0.63686) -- (0.01937,-0.65269) -- (0.06602,-0.66773) -- (0.11260,-0.68196) -- (0.15904,-0.69536) -- (0.20528,-0.70790) -- (0.25128,-0.71959) -- (0.29696,-0.73040) -- (0.34229,-0.74032) -- (0.38720,-0.74933) -- (0.43164,-0.75744) -- (0.47555,-0.76462) -- (0.51888,-0.77087) -- (0.56158,-0.77618) -- (0.60360,-0.78054) -- (0.64488,-0.78396) -- (0.68537,-0.78642) -- (0.72503,-0.78792) -- (0.76381,-0.78846) -- (0.80165,-0.78804) -- (0.83852,-0.78666) -- (0.87437,-0.78432) -- (0.90915,-0.78102) -- (0.94283,-0.77678) -- (0.97535,-0.77159) -- (1.00669,-0.76545) -- (1.03680,-0.75839) -- (1.06565,-0.75040) -- (1.09320,-0.74150) -- (1.11942,-0.73169) -- (1.14427,-0.72099) -- (1.16773,-0.70942) -- (1.18977,-0.69698) -- (1.21035,-0.68369) -- (1.22947,-0.66957) -- (1.24708,-0.65463) -- (1.26318,-0.63889) -- (1.27773,-0.62238) -- (1.29073,-0.60510) -- (1.30216,-0.58709) -- (1.31200,-0.56837) -- (1.31632,-0.55875);
   \draw[draw=black!43,line width=.4pt,opacity=.65] (-0.41493,1.36789) -- (-0.43347,1.36096) -- (-0.45104,1.35339) -- (-0.46763,1.34521) -- (-0.48322,1.33641) -- (-0.49777,1.32701) -- (-0.51129,1.31703) -- (-0.52374,1.30646) -- (-0.53512,1.29534) -- (-0.54542,1.28366) -- (-0.55461,1.27145) -- (-0.56269,1.25871) -- (-0.56965,1.24547) -- (-0.57548,1.23174) -- (-0.58017,1.21754) -- (-0.58372,1.20288) -- (-0.58612,1.18778) -- (-0.58737,1.17227) -- (-0.58746,1.15635) -- (-0.58641,1.14005) -- (-0.58420,1.12339) -- (-0.58085,1.10639) -- (-0.57636,1.08907) -- (-0.57072,1.07145) -- (-0.56395,1.05356) -- (-0.55606,1.03541) -- (-0.54706,1.01702) -- (-0.53695,0.99842) -- (-0.52576,0.97964) -- (-0.51348,0.96068) -- (-0.50015,0.94159) -- (-0.48577,0.92238) -- (-0.47036,0.90307) -- (-0.45394,0.88368) -- (-0.43653,0.86425) -- (-0.41815,0.84480) -- (-0.39883,0.82534) -- (-0.37859,0.80590) -- (-0.35745,0.78651) -- (-0.33543,0.76719) -- (-0.31257,0.74796) -- (-0.28890,0.72885) -- (-0.26444,0.70988) -- (-0.23921,0.69107) -- (-0.21326,0.67245) -- (-0.18662,0.65404) -- (-0.15931,0.63585) -- (-0.13137,0.61792) -- (-0.10283,0.60027) -- (-0.07374,0.58291) -- (-0.04411,0.56587) -- (-0.01400,0.54917) -- (0.01657,0.53282) -- (0.04755,0.51686) -- (0.07891,0.50129) -- (0.11061,0.48614) -- (0.14262,0.47143) -- (0.17488,0.45716) -- (0.20737,0.44337) -- (0.24004,0.43007) -- (0.27286,0.41727) -- (0.30578,0.40499) -- (0.33876,0.39325) -- (0.37177,0.38205) -- (0.40476,0.37141) -- (0.43770,0.36135) -- (0.47053,0.35188) -- (0.50323,0.34301) -- (0.53575,0.33475) -- (0.56806,0.32710) -- (0.60011,0.32009) -- (0.63187,0.31371) -- (0.66329,0.30798) -- (0.69434,0.30291) -- (0.72498,0.29849) -- (0.75517,0.29473) -- (0.78488,0.29165) -- (0.81407,0.28923) -- (0.84271,0.28749) -- (0.87075,0.28643) -- (0.89817,0.28605) -- (0.92493,0.28635) -- (0.95100,0.28732) -- (0.97635,0.28898) -- (1.00094,0.29131) -- (1.02475,0.29431) -- (1.04775,0.29798) -- (1.06991,0.30231) -- (1.09120,0.30731) -- (1.11160,0.31296) -- (1.13108,0.31925) -- (1.14962,0.32619) -- (1.16720,0.33375) -- (1.18378,0.34194) -- (1.19937,0.35074) -- (1.21392,0.36013) -- (1.22744,0.37012) -- (1.23989,0.38068) -- (1.25128,0.39181) -- (1.26157,0.40349) -- (1.27076,0.41570) -- (1.27884,0.42843) -- (1.28580,0.44168) -- (1.29163,0.45541) -- (1.29632,0.46961) -- (1.29987,0.48427) -- (1.30227,0.49936) -- (1.30352,0.51488) -- (1.30361,0.53080) -- (1.30256,0.54709) -- (1.30036,0.56375) -- (1.29700,0.58075) -- (1.29251,0.59807) -- (1.28687,0.61569) -- (1.28011,0.63359) -- (1.27221,0.65174) -- (1.26321,0.67013) -- (1.25926,0.67761);
   \draw[draw=black!43,line width=.4pt,opacity=.65] (-0.38733,1.37655) -- (-0.39545,1.37419) -- (-0.41493,1.36789);
   \draw[draw=black!43,line width=.4pt,opacity=.65] (-0.46996,1.35057) -- (-0.43461,1.36211) -- (-0.39513,1.37326) -- (-0.35517,1.38273) -- (-0.31477,1.39051) -- (-0.27399,1.39660) -- (-0.23288,1.40099) -- (-0.19149,1.40368) -- (-0.14986,1.40465) -- (-0.10805,1.40391) -- (-0.06610,1.40146) -- (-0.02408,1.39731) -- (0.01797,1.39145) -- (0.06000,1.38389) -- (0.10196,1.37465) -- (0.14379,1.36374) -- (0.18545,1.35116) -- (0.22688,1.33694) -- (0.26804,1.32109) -- (0.30887,1.30363) -- (0.34932,1.28458) -- (0.38935,1.26397) -- (0.42890,1.24181) -- (0.46793,1.21814) -- (0.50640,1.19299) -- (0.54424,1.16639) -- (0.58142,1.13836) -- (0.61789,1.10895) -- (0.65361,1.07819) -- (0.68853,1.04611) -- (0.72262,1.01276) -- (0.75582,0.97817) -- (0.78810,0.94239) -- (0.81943,0.90547) -- (0.84975,0.86744) -- (0.87904,0.82835) -- (0.90726,0.78826) -- (0.93437,0.74720) -- (0.96035,0.70524) -- (0.98515,0.66241) -- (1.00875,0.61878) -- (1.03113,0.57440) -- (1.05225,0.52931) -- (1.07208,0.48358) -- (1.09062,0.43726) -- (1.10782,0.39041) -- (1.12367,0.34308) -- (1.13815,0.29533) -- (1.15125,0.24723) -- (1.16294,0.19882) -- (1.17322,0.15017) -- (1.18207,0.10134) -- (1.18947,0.05238) -- (1.19543,0.00336) -- (1.19993,-0.04566) -- (1.20297,-0.09463) -- (1.20455,-0.14348) -- (1.20466,-0.19216) -- (1.20329,-0.24060) -- (1.20047,-0.28875) -- (1.19618,-0.33655) -- (1.19043,-0.38394) -- (1.18323,-0.43086) -- (1.17459,-0.47726) -- (1.16452,-0.52307) -- (1.15304,-0.56825) -- (1.14014,-0.61273) -- (1.12586,-0.65647) -- (1.11020,-0.69941) -- (1.09320,-0.74150) -- (1.07486,-0.78268) -- (1.05521,-0.82291) -- (1.03428,-0.86214) -- (1.01208,-0.90031) -- (0.98865,-0.93739) -- (0.96402,-0.97333) -- (0.93822,-1.00808) -- (0.91127,-1.04161) -- (0.88321,-1.07386) -- (0.85407,-1.10481) -- (0.82390,-1.13441) -- (0.79272,-1.16263) -- (0.76057,-1.18943) -- (0.72750,-1.21478) -- (0.69354,-1.23865) -- (0.65874,-1.26102) -- (0.62313,-1.28184) -- (0.58676,-1.30111) -- (0.54968,-1.31879) -- (0.51193,-1.33486) -- (0.47356,-1.34931) -- (0.46996,-1.35057);
   \draw[draw=black!43,line width=.4pt,opacity=.65] (-1.36203,-0.43562) -- (-1.36369,-0.43038) -- (-1.37560,-0.38957) -- (-1.38582,-0.34828) -- (-1.39436,-0.30657) -- (-1.40120,-0.26448) -- (-1.40633,-0.22207) -- (-1.40975,-0.17939) -- (-1.41145,-0.13650) -- (-1.41143,-0.09343) -- (-1.40969,-0.05026) -- (-1.40624,-0.00702) -- (-1.40107,0.03623) -- (-1.39419,0.07943) -- (-1.38562,0.12254) -- (-1.37535,0.16549) -- (-1.36342,0.20825) -- (-1.34982,0.25075) -- (-1.33457,0.29294) -- (-1.31770,0.33478) -- (-1.29923,0.37621) -- (-1.27917,0.41719) -- (-1.25755,0.45765) -- (-1.23440,0.49756) -- (-1.20975,0.53686) -- (-1.18362,0.57550) -- (-1.15605,0.61345) -- (-1.12707,0.65065) -- (-1.09672,0.68705) -- (-1.06504,0.72262) -- (-1.03205,0.75731) -- (-0.99781,0.79107) -- (-0.96235,0.82387) -- (-0.92572,0.85567) -- (-0.88797,0.88643) -- (-0.84913,0.91610) -- (-0.80925,0.94466) -- (-0.76839,0.97207) -- (-0.72659,0.99829) -- (-0.68391,1.02330) -- (-0.64040,1.04706) -- (-0.59610,1.06954) -- (-0.55108,1.09073) -- (-0.50539,1.11058) -- (-0.45908,1.12908) -- (-0.41221,1.14620) -- (-0.36484,1.16193) -- (-0.31703,1.17624) -- (-0.26883,1.18912) -- (-0.22030,1.20055) -- (-0.17150,1.21052) -- (-0.12250,1.21901) -- (-0.07334,1.22602) -- (-0.02410,1.23154) -- (0.02518,1.23555) -- (0.07442,1.23806) -- (0.12357,1.23906) -- (0.17257,1.23855) -- (0.22136,1.23653) -- (0.26989,1.23300) -- (0.31808,1.22798) -- (0.36588,1.22145) -- (0.41324,1.21344) -- (0.46010,1.20395) -- (0.50640,1.19299) -- (0.55207,1.18058) -- (0.59708,1.16673) -- (0.64136,1.15146) -- (0.68486,1.13479) -- (0.72752,1.11673) -- (0.76929,1.09732) -- (0.81013,1.07656) -- (0.84999,1.05450) -- (0.88880,1.03115) -- (0.92654,1.00654) -- (0.96314,0.98071) -- (0.99857,0.95368) -- (1.03279,0.92550) -- (1.06574,0.89618) -- (1.09740,0.86577) -- (1.12772,0.83431) -- (1.15667,0.80183) -- (1.18421,0.76837) -- (1.21030,0.73398) -- (1.23492,0.69869) -- (1.25804,0.66255) -- (1.27962,0.62561) -- (1.29965,0.58790) -- (1.31809,0.54948) -- (1.33492,0.51038) -- (1.35013,0.47067) -- (1.36203,0.43562);
   \draw[draw=black!43,line width=.4pt,opacity=.65] (-0.80271,-1.18345) -- (-0.80876,-1.17919) -- (-0.82371,-1.16718) -- (-0.83765,-1.15376) -- (-0.85057,-1.13892) -- (-0.86246,-1.12270) -- (-0.87330,-1.10512) -- (-0.88307,-1.08618) -- (-0.89177,-1.06592) -- (-0.89938,-1.04437) -- (-0.90589,-1.02154) -- (-0.91130,-0.99746) -- (-0.91560,-0.97218) -- (-0.91878,-0.94570) -- (-0.92085,-0.91808) -- (-0.92179,-0.88933) -- (-0.92161,-0.85950) -- (-0.92031,-0.82863) -- (-0.91788,-0.79675) -- (-0.91434,-0.76389) -- (-0.90969,-0.73010) -- (-0.90392,-0.69543) -- (-0.89706,-0.65991) -- (-0.88910,-0.62358) -- (-0.88005,-0.58649) -- (-0.86994,-0.54869) -- (-0.85877,-0.51022) -- (-0.84655,-0.47113) -- (-0.83329,-0.43147) -- (-0.81903,-0.39128) -- (-0.80376,-0.35061) -- (-0.78752,-0.30952) -- (-0.77031,-0.26804) -- (-0.75217,-0.22625) -- (-0.73311,-0.18417) -- (-0.71316,-0.14187) -- (-0.69234,-0.09940) -- (-0.67068,-0.05681) -- (-0.64819,-0.01415) -- (-0.62492,0.02853) -- (-0.60089,0.07117) -- (-0.57613,0.11373) -- (-0.55066,0.15615) -- (-0.52452,0.19838) -- (-0.49775,0.24036) -- (-0.47036,0.28206) -- (-0.44241,0.32341) -- (-0.41391,0.36436) -- (-0.38491,0.40488) -- (-0.35545,0.44489) -- (-0.32555,0.48437) -- (-0.29525,0.52326) -- (-0.26459,0.56151) -- (-0.23361,0.59907) -- (-0.20235,0.63591) -- (-0.17083,0.67197) -- (-0.13912,0.70721) -- (-0.10723,0.74159) -- (-0.07521,0.77507) -- (-0.04310,0.80760) -- (-0.01093,0.83915) -- (0.02124,0.86967) -- (0.05339,0.89914) -- (0.08548,0.92751) -- (0.11746,0.95475) -- (0.14930,0.98083) -- (0.18096,1.00571) -- (0.21239,1.02937) -- (0.24357,1.05177) -- (0.27445,1.07290) -- (0.30500,1.09271) -- (0.33517,1.11119) -- (0.36494,1.12832) -- (0.39426,1.14408) -- (0.42310,1.15844) -- (0.45143,1.17139) -- (0.47920,1.18291) -- (0.50640,1.19299) -- (0.53297,1.20162) -- (0.55890,1.20878) -- (0.58414,1.21448) -- (0.60867,1.21869) -- (0.63247,1.22141) -- (0.65549,1.22265) -- (0.67771,1.22240) -- (0.69911,1.22066) -- (0.71965,1.21743) -- (0.73932,1.21272) -- (0.75809,1.20653) -- (0.77593,1.19888) -- (0.79283,1.18976) -- (0.80271,1.18345);
   \draw[draw=black!43,line width=.4pt,opacity=.65] (-0.35213,-1.38597) -- (-0.33993,-1.38848) -- (-0.32439,-1.39009) -- (-0.30846,-1.39000) -- (-0.29215,-1.38822) -- (-0.27548,-1.38474) -- (-0.25848,-1.37958) -- (-0.24116,-1.37274) -- (-0.22355,-1.36423) -- (-0.20567,-1.35405) -- (-0.18754,-1.34223) -- (-0.16917,-1.32877) -- (-0.15061,-1.31369) -- (-0.13185,-1.29700) -- (-0.11294,-1.27874) -- (-0.09389,-1.25893) -- (-0.07473,-1.23757) -- (-0.05547,-1.21471) -- (-0.03615,-1.19037) -- (-0.01678,-1.16458) -- (0.00261,-1.13737) -- (0.02199,-1.10878) -- (0.04135,-1.07883) -- (0.06065,-1.04757) -- (0.07989,-1.01503) -- (0.09902,-0.98126) -- (0.11804,-0.94629) -- (0.13691,-0.91017) -- (0.15561,-0.87294) -- (0.17413,-0.83465) -- (0.19243,-0.79533) -- (0.21050,-0.75505) -- (0.22831,-0.71385) -- (0.24585,-0.67179) -- (0.26308,-0.62890) -- (0.28000,-0.58524) -- (0.29657,-0.54088) -- (0.31278,-0.49585) -- (0.32861,-0.45022) -- (0.34404,-0.40404) -- (0.35905,-0.35737) -- (0.37362,-0.31026) -- (0.38774,-0.26278) -- (0.40139,-0.21497) -- (0.41454,-0.16690) -- (0.42720,-0.11863) -- (0.43933,-0.07022) -- (0.45092,-0.02172) -- (0.46197,0.02681) -- (0.47245,0.07530) -- (0.48236,0.12370) -- (0.49168,0.17195) -- (0.50040,0.22000) -- (0.50851,0.26777) -- (0.51601,0.31522) -- (0.52287,0.36228) -- (0.52910,0.40890) -- (0.53468,0.45503) -- (0.53961,0.50060) -- (0.54388,0.54556) -- (0.54749,0.58985) -- (0.55044,0.63343) -- (0.55271,0.67623) -- (0.55431,0.71822) -- (0.55523,0.75932) -- (0.55548,0.79950) -- (0.55505,0.83871) -- (0.55394,0.87689) -- (0.55216,0.91401) -- (0.54971,0.95001) -- (0.54659,0.98486) -- (0.54280,1.01850) -- (0.53835,1.05091) -- (0.53324,1.08203) -- (0.52749,1.11184) -- (0.52109,1.14029) -- (0.51405,1.16735) -- (0.50640,1.19299) -- (0.49812,1.21718) -- (0.48924,1.23988) -- (0.47976,1.26108) -- (0.46969,1.28073) -- (0.45906,1.29883) -- (0.44786,1.31534) -- (0.43612,1.33025) -- (0.42385,1.34354) -- (0.41106,1.35520) -- (0.39777,1.36520) -- (0.38399,1.37354) -- (0.36975,1.38020) -- (0.35506,1.38519) -- (0.35213,1.38597);
   \draw[regional!65,line width=.6pt] (0.22831,-0.71385) -- (0.26518,-0.72296) -- (0.30184,-0.73150) -- (0.33826,-0.73947) -- (0.37441,-0.74685) -- (0.41027,-0.75364) -- (0.44581,-0.75985) -- (0.48100,-0.76545) -- (0.51581,-0.77045) -- (0.55021,-0.77485) -- (0.58418,-0.77864) -- (0.61769,-0.78181) -- (0.65071,-0.78437) -- (0.68322,-0.78631) -- (0.71520,-0.78763) -- (0.74661,-0.78833) -- (0.77743,-0.78842) -- (0.80764,-0.78788) -- (0.83722,-0.78672) -- (0.86614,-0.78495) -- (0.89438,-0.78255) -- (0.92191,-0.77954) -- (0.94872,-0.77592) -- (0.97478,-0.77169) -- (1.00008,-0.76685) -- (1.02458,-0.76140) -- (1.04829,-0.75536) -- (1.07116,-0.74872) -- (1.09320,-0.74150) -- (1.09320,-0.74150) -- (1.10564,-0.71113) -- (1.11738,-0.68030) -- (1.12842,-0.64898) -- (1.13876,-0.61720) -- (1.14839,-0.58493) -- (1.15731,-0.55217) -- (1.16550,-0.51893) -- (1.17296,-0.48519) -- (1.17968,-0.45094) -- (1.18565,-0.41619) -- (1.19084,-0.38091) -- (1.19525,-0.34510) -- (1.19885,-0.30874) -- (1.20163,-0.27182) -- (1.20355,-0.23432) -- (1.20460,-0.19623) -- (1.20473,-0.15752) -- (1.20392,-0.11816) -- (1.20211,-0.07814) -- (1.19928,-0.03742) -- (1.19536,0.00404) -- (1.19030,0.04628) -- (1.18402,0.08934) -- (1.17645,0.13328) -- (1.16749,0.17816) -- (1.15703,0.22405) -- (1.14495,0.27105) -- (1.13108,0.31925) -- (1.13108,0.31925) -- (1.11550,0.31415) -- (1.09933,0.30945) -- (1.08257,0.30518) -- (1.06524,0.30133) -- (1.04735,0.29791) -- (1.02892,0.29491) -- (1.00997,0.29235) -- (0.99050,0.29022) -- (0.97053,0.28853) -- (0.95008,0.28728) -- (0.92917,0.28646) -- (0.90780,0.28608) -- (0.88601,0.28614) -- (0.86380,0.28663) -- (0.84119,0.28757) -- (0.81820,0.28894) -- (0.79485,0.29075) -- (0.77115,0.29300) -- (0.74713,0.29567) -- (0.72280,0.29878) -- (0.69819,0.30232) -- (0.67331,0.30628) -- (0.64818,0.31067) -- (0.62283,0.31547) -- (0.59726,0.32069) -- (0.57151,0.32632) -- (0.54559,0.33236) -- (0.51952,0.33880) -- (0.51952,0.33880) -- (0.51230,0.29127) -- (0.50470,0.24490) -- (0.49677,0.19959) -- (0.48852,0.15527) -- (0.47999,0.11184) -- (0.47118,0.06927) -- (0.46212,0.02748) -- (0.45282,-0.01356) -- (0.44329,-0.05389) -- (0.43355,-0.09354) -- (0.42361,-0.13255) -- (0.41346,-0.17094) -- (0.40312,-0.20874) -- (0.39261,-0.24596) -- (0.38191,-0.28263) -- (0.37104,-0.31876) -- (0.36000,-0.35435) -- (0.34880,-0.38944) -- (0.33743,-0.42402) -- (0.32591,-0.45811) -- (0.31424,-0.49172) -- (0.30241,-0.52485) -- (0.29043,-0.55750) -- (0.27830,-0.58969) -- (0.26602,-0.62142) -- (0.25360,-0.65269) -- (0.24103,-0.68350) -- (0.22831,-0.71385) -- cycle;
   \path[fill=sphereLight,draw=regional!55,line width=.55pt,opacity=.65,rotate around={-131.59641:(-0.71394,-0.80423)}] (-0.71394,-0.80423) ellipse [x radius=1.1864pt,y radius=1.8000pt];
   \path[fill=sphereLight,draw=regional!55,line width=.55pt,opacity=.65,rotate around={-53.19251:(0.65659,-0.87744)}] (0.65659,-0.87744) ellipse [x radius=1.1563pt,y radius=1.8000pt];
   \path[fill=sphereLight,draw=regional!55,line width=.55pt,opacity=.65,rotate around={111.54223:(-0.42885,1.08634)}] (-0.42885,1.08634) ellipse [x radius=1.0386pt,y radius=1.8000pt];
   \path[fill=sphereLight,draw=regional!55,line width=.55pt,opacity=.65,rotate around={-86.82875:(0.07546,-1.36190)}] (0.07546,-1.36190) ellipse [x radius=0.5406pt,y radius=1.8000pt];
   \path[fill=sphereLight,draw=regional!55,line width=.55pt,opacity=.65,rotate around={-128.35916:(-0.86882,-1.09778)}] (-0.86882,-1.09778) ellipse [x radius=0.3668pt,y radius=1.8000pt];
   \path[fill=sphereLight,draw=regional!55,line width=.55pt,opacity=.65,rotate around={74.30188:(0.37915,1.34903)}] (0.37915,1.34903) ellipse [x radius=0.3588pt,y radius=1.8000pt];
   \path[fill=sphereLight,draw=regional!55,line width=.55pt,opacity=.65,rotate around={-164.75374:(-1.36297,-0.37149)}] (-1.36297,-0.37149) ellipse [x radius=0.2792pt,y radius=1.8000pt];
   \path[fill=sphereLight,draw=regional!55,line width=.55pt,opacity=.65,rotate around={85.54760:(0.11059,1.42025)}] (0.11059,1.42025) ellipse [x radius=0.1570pt,y radius=1.8000pt];
   \path[fill=sphereLight,draw=regional!55,line width=.55pt,opacity=.65,rotate around={121.68595:(-0.75025,1.21542)}] (-0.75025,1.21542) ellipse [x radius=0.0870pt,y radius=1.8000pt];
   \path[fill=regional,draw=white,line width=.38pt,rotate around={23.36774:(1.23603,0.53405)}] (1.23603,0.53405) ellipse [x radius=0.6735pt,y radius=2.0000pt];
   \path[fill=regional,draw=white,line width=.38pt,rotate around={-117.66399:(-0.57698,-1.10066)}] (-0.57698,-1.10066) ellipse [x radius=0.9895pt,y radius=2.0000pt];
   \path[fill=regional,draw=white,line width=.38pt,rotate around={-179.97788:(-1.22370,-0.00047)}] (-1.22370,-0.00047) ellipse [x radius=1.0348pt,y radius=2.0000pt];
   \path[fill=regional,draw=white,line width=.38pt,rotate around={94.07678:(-0.07951,1.11549)}] (-0.07951,1.11549) ellipse [x radius=1.2464pt,y radius=2.0000pt];
   \path[fill=regional,draw=white,line width=.38pt,rotate around={67.06560:(0.39216,0.92682)}] (0.39216,0.92682) ellipse [x radius=1.4209pt,y radius=2.0000pt];
   \path[fill=regional,draw=white,line width=.38pt,rotate around={-34.81981:(0.73883,-0.51388)}] (0.73883,-0.51388) ellipse [x radius=1.5542pt,y radius=2.0000pt];
   \path[fill=regional,draw=white,line width=.38pt,rotate around={-95.34989:(-0.03624,-0.38695)}] (-0.03624,-0.38695) ellipse [x radius=1.9247pt,y radius=2.0000pt];
   \node[text=residual,font=\Large,inner sep=0pt] at (-0.33462,0.74039) {$\star$};
  \end{scope}
  \draw[black!45,line width=.5pt] (0,0) circle (1.43);
  \draw[-{Latex[length=1.1mm,width=.8mm]},regional!80,line width=.45pt]
    (.97,.05) -- (1.49,.60);
  \node[text=regional,anchor=west,inner sep=1pt] at (1.52,.62) {$C_j$};

  \draw[regional!70,line width=.45pt] (0.73883,-0.51388) -- (1.49600,-0.36388);
  \node[text=regional,anchor=west,inner sep=1pt] at (1.51300,-0.36388) {$\omega_j$};
  \node[text=residual,anchor=west,inner sep=1pt] at (-0.21462,0.78039) {$\omega_N$};
 \end{scope}
 \node[card,draw=black!25,fill=black!2,minimum width=4.08cm,
       minimum height=.8cm] (sharedlist) at (2.19,1.77) {};
 \foreach \x/\lab in {.65/\omega_1,1.42/\omega_2,2.19/\cdots,2.96/\omega_m} {
   \node[card,draw=regional!55,fill=regional!10,
         minimum width=.68cm,minimum height=.56cm] at (\x,1.77) {$\lab$};
 }
 \node[card,draw=residual!65,fill=residual!12,
       minimum width=.68cm,minimum height=.56cm] at (3.73,1.77) {$\omega_N$};
 \node[text=black!75] at (2.15,1.17) {$N=m+1$};
 \node (sharedalice) at (1.05,.56) {Alice};
 \node (sharedbob) at (3.25,.56) {Bob};
 \draw[arr] (sharedalice.north |- sharedlist.south) -- (sharedalice.north);
 \draw[arr] (sharedbob.north |- sharedlist.south) -- (sharedbob.north);

 \begin{scope}[shift={(4.65,0)}]
  \node[heading] at (0,5.68) {(b) Exact local selection};
  \node at (2.85,5.12) {$f_x=\underline f_x+h_x$};
  \begin{scope}[shift={(.47,2.55)}]
   \foreach \j in {0,...,5} {
    \pgfmathsetmacro{\xa}{.8*\j}
    \pgfmathsetmacro{\xb}{.8*(\j+1)}
    \pgfmathsetmacro{\g}{1.55*(.5+.36*sin(75*(\xa+.4)))}
    \fill[regional!23] (\xa,0) rectangle (\xb,\g);
    \fill[residual!24] (\xa,\g) --
      plot[domain=\xa:\xb,samples=21] (\x,{1.55*(1+.36*sin(75*\x))})
      -- (\xb,\g) -- cycle;
    \draw[regional,line width=.65pt] (\xa,\g) -- (\xb,\g);
    \ifnum\j>0
     \pgfmathsetmacro{\gprev}{1.55*(.5+.36*sin(75*(\xa-.4)))}
     \draw[regional,line width=.65pt] (\xa,\gprev) -- (\xa,\g);
    \fi
    \draw[black!20,densely dotted] (\xa,0) --
      (\xa,{1.55*(1+.36*sin(75*\xa))});
   }
   \draw[black!85,line width=.8pt] plot[domain=0:4.8,samples=101] (\x,{1.55*(1+.36*sin(75*\x))});
   \draw[arr,black!65] (-.08,0) -- (5.02,0);
   \draw[arr,black!65] (0,-.02) -- (0,2.35);
   \node[text=regional] at (1.15,.62) {$\underline f_x$};
   \node[text=residual] at (1.15,1.7) {$h_x$};
   \node[anchor=west] at (4.83,1.55) {$f_x$};
  \end{scope}
  \node[card,draw=residual!65,fill=residual!9,
        minimum width=5.05cm,minimum height=.72cm] (globalchoice) at (2.85,1.85) {};
  \node[text=residual!90!black] at (.92,1.85) {$\omega_N$};
  \draw[residual!35] (1.48,1.49) -- (1.48,2.21);
  \node at (3.4,1.85) {with probability $h_x(\omega_N)$};
  \node[card,draw=regional!65,fill=regional!8,
        minimum width=5.05cm,minimum height=.72cm] (regionalchoice) at (2.85,.63) {};
  \node[text=regional!90!black] at (.92,.63) {$\omega_j$};
  \draw[regional!35] (1.48,.27) -- (1.48,.99);
  \node at (3.4,.63) {with probability $w_j(x)/W_x$};
  \draw[arr] (globalchoice.south) --
    node[midway,right=2pt,inner sep=0pt] {otherwise} (regionalchoice.north);
 \end{scope}

 \begin{scope}[shift={(10.95,0)}]
  \node[heading] at (0,5.68) {(c) Nonsignalling coordination};
  \node[inner sep=0pt] (x) at (1.4,5.32) {$x$};
  \node[inner sep=0pt] (y) at (4.85,5.32) {$y$};
  \node[card,fill=regional!4,minimum width=2.5cm,
        minimum height=1.32cm] (alice) at (1.4,4.26) {};
  \node[card,fill=regional!4,minimum width=2.5cm,
        minimum height=1.32cm] (bob) at (4.85,4.26) {};
  \node[font=\footnotesize\bfseries] at (1.4,4.64) {Alice};
  \node[font=\footnotesize\bfseries] at (4.85,4.64) {Bob};
  \foreach \dx/\lab in {-.81/1,-.27/2,.27/\cdots,.81/N} {
    \node[card,minimum width=.48cm,minimum height=.44cm,
          font=\scriptsize] at (1.4+\dx,4.23) {$\lab$};
  }
  \node[card,draw=regional,line width=1pt,fill=regional!25,
        minimum width=.48cm,minimum height=.44cm,font=\scriptsize]
        at (1.13,4.23) {$2$};
  \node[font=\scriptsize,inner sep=1pt] (signbit) at (1.4,3.81)
    {$(-1)^d=a_x(\omega_I)$};
  \node[card,minimum width=2.22cm,minimum height=.44cm,
        font=\scriptsize] at (4.85,4.23) {$\mathbf z=(z_1,\ldots,z_N)$};
  \draw[arr] (x.south) -- (alice.north);
  \draw[arr] (y.south) -- (bob.north);
  \node[card,fill=black!2,minimum width=5.7cm,
        minimum height=1.35cm] (indexbox) at (3.15,2.44) {};
  \node at (3.15,2.84) {\textbf{Index box} $R_N$};
  \foreach \px/\lab in {1.1/1,2.2/2,4.3/N-1} {
    \node[card,draw=black!45,fill=white,minimum width=.92cm,
          minimum height=.43cm,font=\scriptsize] at (\px,2.24) {$\mathrm{PR}_{\lab}$};
  }
  \node at (3.25,2.24) {$\cdots$};
  \node[circle,draw=black!50,fill=white,line width=.5pt,
        minimum size=.43cm,inner sep=1pt,font=\scriptsize] at (5.45,2.24) {$H$};
  \draw[arr] (alice.south) -- node[midway,left=2pt,inner sep=0pt] {$I$} (alice.south |- indexbox.north);
  \draw[arr] (bob.south) -- node[midway,right=2pt,inner sep=0pt] {$\mathbf z$} (bob.south |- indexbox.north);
  \node[card,minimum width=2.35cm,minimum height=.6cm] (ao)
     at (1.4,.57) {$A=(-1)^{u\oplus d}$};
  \node[card,minimum width=2.35cm,minimum height=.6cm] (bo)
     at (4.85,.57) {$B=(-1)^v$};
  \draw[arr] (alice.south |- indexbox.south) -- node[midway,left=2pt,inner sep=0pt] {$u$} (ao.north);
  \draw[arr] (bob.south |- indexbox.south) -- node[midway,right=2pt,inner sep=0pt] {$v$} (bo.north);
  \node[font=\scriptsize] at (3.15,1.32) {Output parity\\$u\oplus v=z_I$};
  \draw[arr,dashed,rounded corners=2pt] (signbit.west) -- (-.08,3.81)
    -- (-.08,.57) -- (ao.west);
 \end{scope}
\end{tikzpicture}
\caption{Exact simulation of the nonlocal branch (probability $1-c$).
(a) The highlighted front patch and its antipodal back patch together form
$C_j$. Each such pair supplies one independent, area-uniform blue candidate
$\omega_j$. Dark filled and pale outlined markers indicate front and back
samples, respectively. The independent orange
star $\omega_N$ is uniform over the sphere. The grid is illustrative.
(b) Schematic decomposition into a regional lower envelope (blue) and a
residual (orange). Alice selects $\omega_N$ with probability $h_x(\omega_N)$,
otherwise $\omega_j$ with conditional probability $w_j(x)/W_x$.
(c) Alice's index $I$ and Bob's full table $\mathbf z$ enter $R_N$, which
returns individually fair bits $u,v$ with parity $z_I$. Inside $R_N$, rectangles
denote $N-1$ parallel PR calls and the circle a shared fair bit $H$.
The dashed path carries Alice's local sign bit $d$ to her output.}
\label{fig:protocol}
\end{figure*}
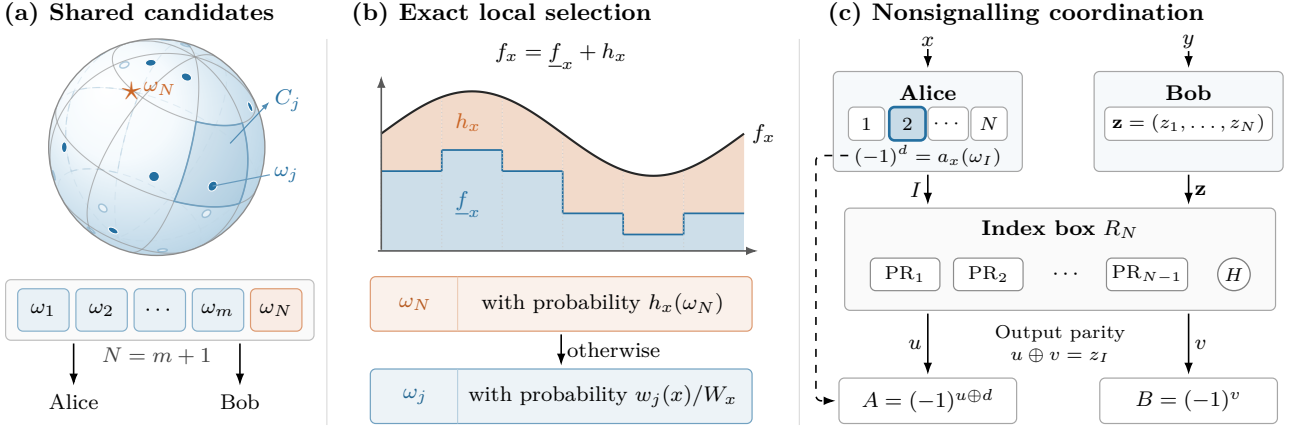

\mainheading{\textit{Nonlocal branch}}
We now describe the new protocol for the \textit{nonlocal branch}.
Figure~\ref{fig:protocol} shows its three ingredients: a shared list prepared
before the inputs arrive, an exact local selection, and nonsignalling
coordination of the outputs.

\protocolstep{1}{Preparation}
Before receiving their measurement settings, Alice and Bob use shared
randomness to prepare the same finite list of candidate directions.
They will later use this list differently: Alice selects one candidate,
while Bob evaluates his response on every candidate.

The sampling rule is fixed in advance for the known state. Partition
the sphere into $m$ regions $C_1,\ldots,C_m$, each of positive area.
For each region $C_j$, both parties know its probability mass $\mu(C_j)$
and a fixed reference direction $v_j\in C_j$. Alice uses the references to
calculate selection probabilities. The candidates are resampled independently for each run.

Specifically, the parties draw one direction $\Lambda_j$ uniformly within
each region $C_j$ ($j=1,\ldots,m$), and one additional direction $\Lambda_N$
uniformly over the whole sphere:
\begin{equation}
 \Lambda_j\sim\mu(\cdot\mid C_j)\quad(j\le m),\qquad
 \Lambda_N\sim\mu,\qquad N=m+1.
 \label{eq:candidates}
\end{equation}
Here $\mu(\cdot\mid C_j)=\mu(\cdot\cap C_j)/\mu(C_j)$ is the uniform
probability measure restricted to region $C_j$ and normalized to mass one.
All draws are independent. Write the random list as
$\Omega=(\Lambda_1,\ldots,\Lambda_N)$, its joint distribution as $Q$,
and a realized list as $\omega=(\omega_1,\ldots,\omega_N)$, where
$\omega_j$ is the sampled value of $\Lambda_j$.

They also share an independent fair bit $H$ and have $N-1$ independent
PR boxes available for the nonsignalling box defined in Step~4. The bit is
independent of the list, branch coin, local selection coins, and boxes.
None of these resources depends on the measurement settings.

The partition and reference directions are chosen so that, throughout
each region, every target density stays within $1/2$ of its value at
the reference direction:
\begin{equation}
 \abs{f_x(\lambda)-f_x(v_j)}\le\frac12,
 \qquad \lambda\in C_j,\quad x\in\sphere.
 \label{eq:cellvariation}
\end{equation}
Joint continuity on the compact product $\sphere\times\sphere$ makes
$f_x$ uniformly continuous in $\lambda$, uniformly over $x$, so a
sufficiently fine finite partition satisfies this bound.
This condition ensures that Alice's selection probabilities below are
well defined. An explicit Borel partition, together with its region
masses and sampling formulas, is given in SM Sec.~\ref{sm:grid}.

\protocolstep{2}{Alice's selection}
Upon receiving $x$, Alice selects an index $I\in\{1,\ldots,N\}$ and hence
the candidate $\omega_I$, without generating new candidates or communicating
her choice.

The idea is to split the target density into a piecewise-constant lower
envelope and a residual bounded by one [Fig.~\ref{fig:protocol}(b)].
The regional candidates supply the lower-envelope contribution, while
the global candidate supplies the residual.

To determine the selection probabilities, Alice evaluates $f_x$ at the
fixed reference directions $v_1,\ldots,v_m$ and at the global candidate
$\omega_N$. From these values, she defines
\begin{align}
 \gamma_j(x)&=\pos{f_x(v_j)-\tfrac12},\nonumber\\
 \underline f_x(\lambda)&=\sum_{j=1}^m\gamma_j(x)\ind\{\lambda\in C_j\},\nonumber\\
 h_x(\lambda)&=f_x(\lambda)-\underline f_x(\lambda).
 \label{eq:lowerresidual}
\end{align}
Here $\gamma_j(x)$ is the nonnegative level assigned to region $C_j$.
Equation~\eqref{eq:cellvariation} ensures that
$\gamma_j(x)\le f_x(\lambda)$ for $\lambda\in C_j$.
Thus the piecewise-constant function $\underline f_x$ is a lower envelope:
$0\le\underline f_x(\lambda)\le f_x(\lambda)$ for every $\lambda$.
The residual $h_x$ is the difference between the target density $f_x$ and
this lower envelope $\underline f_x$. If $\omega_N$ lies in region $C_k$,
she computes $h_x(\omega_N)=f_x(\omega_N)-\gamma_k(x)$. This is the probability with
which she selects the global candidate $\omega_N$. If $\omega_N$ is not selected, she computes the unnormalized weights
\begin{equation}
 w_j(x)=\mu(C_j)\gamma_j(x),\qquad
 W_x=\sum_{j=1}^m w_j(x).
 \label{eq:cellweights}
\end{equation}
The partition condition~\eqref{eq:cellvariation} ensures that
$0\le h_x\le1$ and $W_x>0$.

Alice now uses a local random coin to select $I=N$ with probability
$h_x(\omega_N)$. Otherwise, she draws $j\in\{1,\ldots,m\}$ with
probability $w_j(x)/W_x$ and sets $I=j$. Equivalently, her index has
conditional probabilities
\begin{align}
 p(N\mid x,\omega)&=h_x(\omega_N),\nonumber\\
 p(j\mid x,\omega)&=
 [1-h_x(\omega_N)]\frac{w_j(x)}{W_x},\quad j\le m.
 \label{eq:selector}
\end{align}
Every run selects exactly one candidate and no run is restarted or discarded.

The distribution of the selected vector is obtained by averaging over
Alice's index choices and the shared random list. Selecting the global
candidate $\omega_N$ contributes density $h_x$ relative to $\mu$.
For each regional candidate $\omega_j$, Eq.~\eqref{eq:selector} gives
$\Prob(I=j\mid x)=w_j(x)$, because
$\int[1-h_x(\lambda)]\,\mu(\dd\lambda)=W_x$.
The decision to select $j$ is independent of the regional candidate
$\Lambda_j$, which is uniformly distributed in $C_j$. Thus the regional candidates
$\omega_1,\ldots,\omega_{N-1}$ together contribute density
$\sum_{j=1}^{N-1}w_j(x)\ind\{\lambda\in C_j\}/\mu(C_j)
=\underline f_x(\lambda)$.
Adding the global and regional contributions gives density
$h_x+\underline f_x=f_x$, so for every $x$,
\begin{equation}
 \Omega\sim Q,\quad I\sim p(\cdot\mid x,\Omega)
 \quad\Longrightarrow\quad \Lambda_I\sim\nu_x.
 \label{eq:samplinggoal}
\end{equation}
Conditioned on a particular list $\omega$, the selected vector
can only be one of $\omega_1,\ldots,\omega_N$, whereas averaging this conditional distribution over the shared list reproduces $\nu_x$ exactly. The general finite-selection
result and its proof are given in SM Sec.~\ref{sm:selection}.

Finally, Alice encodes her sign response as
$d=\ind\{(Tx)\cdot\omega_I<0\}$, so that $a_x(\omega_I)=(-1)^d$.
She retains $I$ as her input to the nonsignalling box defined in Step~4,
and $d$ for the final output processing.

\protocolstep{3}{Bob's table}
Upon receiving $y$, Bob computes the sign of $y\cdot\omega_j$ for
every vector in the shared list. He records these signs as bits
$z_j=\ind\{y\cdot\omega_j<0\}$ for $j=1,\ldots,N$, so that
$b_y(\omega_j)=(-1)^{z_j}$. Bob forms the entire table
$z=(z_1,\ldots,z_N)\in\{0,1\}^N$ because he does not know Alice's index.
Using only $y$ and the shared list, he can work in parallel with her selection.

\protocolstep{4}{Index-box call and outputs}
Alice inputs $I$ and Bob his full table $z$ into the index box $R_N$,
which returns a bit $u$ to Alice and a bit $v$ to Bob according to
\begin{equation}
 R_N(u,v\mid i,z)=\frac12\ind\{u\oplus v=z_i\},
 \label{eq:indexbox}
\end{equation}
where $i\in\{1,\ldots,N\}$, $z\in\{0,1\}^N$, and $u,v\in\{0,1\}$.
Each returned bit is individually fair, while their parity equals the
table entry selected by Alice. Neither party receives the other
party's input or output. This nonsignalling random-access
correlation~\cite{Chaturvedi2017} is implemented with parallel PR-box
calls as described in the next section.

To produce her final output, Alice combines $u$ with her previously
computed sign bit $d$ using XOR, then converts the resulting bit to a
sign. Bob converts his bit $v$ directly to a sign:
\begin{equation}
 A=(-1)^{u\oplus d},\qquad B=(-1)^v.
 \label{eq:finaloutputs}
\end{equation}
Both parties then return their outputs, completing the run without
communication. The individually fair box outputs ensure uniform
marginals, while their parity gives
$AB=(-1)^{d+z_I}=a_x(\omega_I)b_y(\omega_I)$.
Averaging over the selected direction therefore recovers $C_N(x,y)$
and hence the full distribution $P_N$. The distributional verification
is given in SM Sec.~\ref{sm:boxcorrectness}.

\mainheading{\textit{PR implementation and budget}}
The box $R_N$ has input alphabet sizes $N$ and $2^N$, with binary outputs.
It can be implemented by $N-1$ independent PR boxes. For each $j<N$, the
parties supply
\begin{equation}
 a_j=\ind\{I=j\},\qquad b_j=z_j\oplus z_N.
 \label{eq:prinputs}
\end{equation}
From the respective outputs $u_j,v_j$ and the shared bit $H$, they form
\begin{equation}
 u=H\oplus\bigoplus_{j<N}u_j,\qquad
 v=H\oplus z_N\oplus\bigoplus_{j<N}v_j.
 \label{eq:proutputs}
\end{equation}
All box inputs are fixed by local preprocessing, so the calls are parallel.
This wiring specializes standard distributed Boolean-function
constructions~\cite{vanDam2013,Kaplan2011}. Its correctness and optimal cost
for the full input domain of $R_N$ are proved in SM Sec.~\ref{sm:index}.

The partition required in Step~1 also fixes the PR budget: $m$ regions give
$N=m+1$ candidates and hence $K=N-1=m$ calls. The construction in
SM Sec.~\ref{sm:grid} satisfies Eq.~\eqref{eq:cellvariation} by pairing
northern angular patches with their antipodes. A grid with angular spacing
$\pi/(2n_\theta)$ suffices when
\begin{equation}
 n_\theta=\left\lceil\sqrt2\pi\cot^2\theta\right\rceil,
 \qquad m=4n_\theta^2,\qquad N=m+1.
 \label{eq:gridbudget}
\end{equation}
Its cell masses and exact sampling formulas are specified in
SM Sec.~\ref{sm:grid}. The selection probabilities are then computed from
these known cell masses, pointwise evaluations of $f_x$, and finite sums.
The resulting $K=4n_\theta^2$ gives
the budget in Theorem~\ref{thm:main}. The finite-box verification and resource
estimates complete its proof in SM Sec.~\ref{sm:mainproof}.

\mainheading{\textit{Mixed states and binary POVMs}}
The finite-resource simulation extends to mixed states
and local binary-outcome measurements using only
shared randomness and local postprocessing.

\begin{corollary}[Mixed states and binary POVMs]
\label{cor:mixed}
For every fixed two-qubit state $\rho$, there is a finite,
state-dependent number of parallel PR-box calls sufficient
to reproduce exactly the statistics of all pairs of local
two-outcome POVMs, without communication or postselection.
\end{corollary}

To see this, choose a finite pure-state decomposition
$\rho=\sum_{r=1}^{R}p_r\ket{\psi_r}\bra{\psi_r}$,
with $R\le4$. The parties use shared randomness to
select $r$ with probability $p_r$, independently of
the measurements, and run the corresponding pure-state
simulator, absorbing the required local changes of Schmidt
basis into their measurement settings.
If $\theta_r\in[0,\pi/4]$ denotes the Schmidt parameter
of $\ket{\psi_r}$, at most
$K_\rho=\max_{r:p_r>0}K(\theta_r)$
PR boxes are needed in any run.

Every binary qubit POVM is a stochastic postprocessing
of a projective measurement. Indeed, for Alice's POVM
$E=\{E_+,E_-\}$, write
$E_+=e_+\Pi_{+|x}+e_-\Pi_{-|x}$,
with $0\le e_\pm\le1$ and $E_-=\id_2-E_+$.
Alice first simulates the projective measurement
along $x$ and, conditional on its outcome
$a\in\{+1,-1\}$, reports $+$ with probability
$e_a$ and $-$ otherwise. If $E_+$ is proportional
to the identity, any $x$ may be used.
Bob applies the analogous procedure to
$F=\{F_+,F_-\}$. By linearity, their final
outputs satisfy
\begin{equation}
 \Prob(A=\alpha,B=\beta\mid E,F)
 =\Tr\!\left[\rho(E_\alpha\otimes F_\beta)\right],
 \label{eq:povmtarget}
\end{equation}
for all local binary POVMs. Neither shared mixing
nor local postprocessing requires additional PR boxes;
see SM Sec.~\ref{sm:extensions}.

\mainheading{\textit{Discussion}}
Our result resolves the finite-alphabet simulation
question for two-qubit correlations: for each fixed
two-qubit state, a continuum of local measurement
settings can be simulated exactly with finitely many
uses of binary-input, binary-output PR boxes.
The essential distinction is between continuous
data handled by shared randomness and local processing,
and the finite nonlocal coordination needed to
reproduce the correlations. The finite-list selector
makes this separation explicit: it converts
input-dependent sampling into a finite-index task
without approximating the target distribution or
communicating the selected index.

The construction provides explicit, though not
necessarily optimal, resource guarantees.
For $0<\theta<\pi/4$, the worst-case cost is
$K(\theta)=O(\theta^{-4})$, while the average number
of PR-box calls in our protocol is
$(1-\cos2\theta)K(\theta)=O(\theta^{-2})$
as $\theta\to0^+$. Their divergence reflects this
construction; it does not imply an intrinsic divergence
of the resources required near a product state.
Moreover, for any fixed $0<\theta_0<\pi/4$,
the same finite partition, candidate-list distribution,
and supply of $K(\theta_0)$ PR boxes suffice for
every $\theta\in[\theta_0,\pi/4)$, with
state-dependent local processing
(SM Sec.~\ref{sm:uniform}).
Whether a single finite budget suffices uniformly
over all pure two-qubit states, and what the optimal
state-dependent PR cost is, remain open.
In particular, maximizing the local weight does
not by itself establish optimality of the PR-box cost.

The finite-selection principle is not restricted
to the particular qubit density used here.
Theorem~\ref{thm:selection} in the SM applies
to jointly continuous families of normalized
densities on compact parameter and sample spaces,
relative to a full-support reference measure.
Combined with Borel-measurable local sign responses,
it yields finite PR-box realizations of unbiased
binary correlations admitting such a
distributed-sampling representation.
Identifying further quantum correlation families
with this structure, and extending the approach
to general multi-outcome POVMs, are natural
directions; the latter is not covered by the
binary postprocessing used above~\cite{Oszmaniec2017}.
Finite-precision implementations introduce a
separate challenge: both sampling errors and
sign evaluations near zero must be controlled.
The stability bound in SM Sec.~\ref{sm:stability}
addresses distributional sampling errors assuming
exact sign responses and index-box operations,
while a complete precision--accuracy--resource
tradeoff remains open.
These results motivate a broader resource question:
which continuous quantum correlation families admit
exact finite-alphabet nonsignalling realizations,
and at what cost?

\mainheading{\textit{Acknowledgments}}
This work was supported by the National Natural Science Foundation of China (Grant. No.~92576114, 12447107), the Guangdong Provincial Quantum Science Strategic Initiative (Grant No.~GDZX2403008, GDZX2503001), and by the Guangdong Basic and Applied Basic Research Foundation (Grant Nos.~2026A1515030035 and 2025A1515110223).

\mainheading{\textit{Use of AI-assisted tools}}
Frontier Large Language Models were used during the development of this work. The authors rigorously verified the statement and proof, and prepared the manuscript. The authors take full responsibility for the final text and all mathematical claims contained therein.

\bibliography{ref}

% Supplemental Material: a fresh page, single-column layout, S-numbering.
\clearpage
\onecolumngrid
\phantomsection
\pdfbookmark[0]{Supplemental Material}{supplement}
\begin{center}
 {\large\bfseries Supplemental Material}\\[5pt]
 {\bfseries Finite Nonsignalling Resources Suffice to Simulate Entangled Qubits}
\end{center}
\medskip
The supplement follows the dependencies of the protocol.
Section~\ref{sm:quantum} verifies the known local--nonlocal representation.
Section~\ref{sm:selection} states and proves the general finite-selection
theorem, and Sec.~\ref{sm:grid} constructs the explicit qubit partition and
candidate distributions. Section~\ref{sm:index} verifies the index box, establishes
the resource bounds, and completes the proof of the main theorem; it also
gives the optimal PR cost of the full index box.
Section~\ref{sm:extensions} proves the mixed-state and binary-measurement
extensions. Section~\ref{sm:stability} treats perturbations of the distribution
of the selected vector separately from the exact construction. Reference numbering matches the main text. All measures and kernels are defined on Borel sigma-algebras.

\appendix
\makeatletter
\@removefromreset{equation}{section}
\renewcommand{\p@subsection}{}
\makeatother
\setcounter{secnumdepth}{2}
\setcounter{section}{0}
\renewcommand{\thesection}{S\arabic{section}}
\renewcommand{\thesubsection}{\thesection.\arabic{subsection}}
\setcounter{equation}{0}
\renewcommand{\theequation}{S\arabic{equation}}
\renewcommand{\theHequation}{supp.\arabic{equation}}
\renewcommand{\theHsection}{supp.\arabic{section}}
\renewcommand{\theHsubsection}{supp.\arabic{section}.\arabic{subsection}}

\section{Qubit representation and regularity}
\label{sm:quantum}

We first verify the decomposition used in the main text, including the
clipping form of the local response. Its source is Ref.~\cite{Renner2023},
Protocol 0 and Appendix E, Protocol 6. Throughout this section $0<\theta<\pi/4$,
$c=\cos(2\theta)$, $s=\sin(2\theta)$, $T=\diag(s,-s,1)$, $t_x=Tx$, and
$e_z=(0,0,1)$. Surface area on $\sphere$ is denoted by $\dd S$.

\subsection{Spherical identities and the local response}

For $w\in\R^3$ and a unit vector $y$,
\begin{equation}
 \int_{\sphere}\abs{w\cdot\lambda}\,\dd S=2\pi\norm{w},
 \qquad
 \int_{\sphere}\lambda\sgn(y\cdot\lambda)\,\dd S=2\pi y.
 \label{eq:sm-sphereidentities}
\end{equation}
The first identity follows by integration in spherical coordinates with
the polar axis along $w$.
The second integral is parallel to $y$ by rotational symmetry, and its
component along $y$ is $\int\abs{y\cdot\lambda}\,\dd S=2\pi$. Also,
\begin{equation}
 \norm{t_x\pm ce_z}^2
 =s^2(1-x_z^2)+(x_z\pm c)^2=(1\pm cx_z)^2,
 \qquad
 \norm{t_x\pm ce_z}=1\pm cx_z.
 \label{eq:sm-schmidtidentity}
\end{equation}
The final expressions are positive because $c<1$.

For each $\lambda$, define $u_x=t_x\cdot\lambda$ and $k=c\abs{\lambda_z}$.
Decompose the scalar response as
\begin{equation}
 u_x=q_x+r_x,\qquad
 q_x=\clip(u_x,-k,k),\qquad
 r_x=\sgn(u_x)\pos{\abs{u_x}-k}.
 \label{eq:sm-clipping}
\end{equation}
On the local branch, the parties share a point on
$H=\{\lambda\in\sphere:\lambda_z\ge0\}$ with area density
$\lambda_z/\pi$. This density is normalized because
$\int_H\lambda_z\,\dd S=\pi$. Bob outputs $b_y(\lambda)$, while Alice's
conditional mean is
\begin{equation}
 m_A(x,\lambda)
 =\clip\!\left(\frac{t_x\cdot\lambda}{c\lambda_z},-1,1\right)
 =\frac{q_x(\lambda)}{c\lambda_z}
 \quad (\lambda_z>0).
 \label{eq:sm-localmean}
\end{equation}
At $\lambda_z=0$ the mean may be fixed arbitrarily, since the equator has
zero area. All response functions are Borel.

To identify this rule with the construction in Ref.~\cite{Renner2023}, set its Schmidt probability
$p=\cos^2\theta$, so $2p-1=c$. Alice's quantum marginal and Bob's conditional
Bloch vector obey
\begin{equation}
 p_\alpha=\frac{1+\alpha cx_z}{2},\qquad
 p_\alpha v_\alpha=\frac{ce_z+\alpha t_x}{2}.
 \label{eq:sm-sourcevectors}
\end{equation}
For $\lambda_z>0$, Protocol 0 of that reference therefore gives
\begin{equation}
 p_A(\alpha\mid x,\lambda)
 =\frac{[c\lambda_z+\alpha u_x]_+}
 {[c\lambda_z+u_x]_++[c\lambda_z-u_x]_+}
 =\frac{1+\alpha m_A(x,\lambda)}2.
 \label{eq:sm-sourceresponse}
\end{equation}
Its local shared-vector density is $[\lambda_z]_+/\pi$, as used here.

Antipodal symmetry and~\eqref{eq:sm-sphereidentities} give
\begin{equation}
 \E B_L=\frac1\pi\int_H\lambda_z b_y(\lambda)\,\dd S
 =\frac1{2\pi}e_z\cdot\int_{\sphere}\lambda b_y(\lambda)\,\dd S=y_z.
 \label{eq:sm-localbobmean}
\end{equation}
Using
$\clip(u,-k,k)=(\abs{u+k}-\abs{u-k})/2$, we obtain
\begin{equation}
 c\E A_L=\frac1\pi\int_Hq_x\,\dd S
 =\frac12\left(\norm{t_x+ce_z}-\norm{t_x-ce_z}\right)=cx_z.
 \label{eq:sm-localfirstmoment}
\end{equation}
In this step, each absolute-value integrand is antipodally even, so its
hemisphere integral is half its full-sphere integral. Thus $\E A_L=x_z$. The product $q_xb_y$ is antipodally even outside
null equators, so the local correlation $C_L(x,y)=\E(A_LB_L\mid x,y)$
satisfies
\begin{equation}
 cC_L(x,y)=\frac1{2\pi}\int_{\sphere}q_x(\lambda)b_y(\lambda)\,\dd S.
 \label{eq:sm-localcorrelation}
\end{equation}

\subsection{Normalization of the remainder and the Born distribution}

For $k\ge0$,
$\pos{\abs{u}-k}=(\abs{u+k}+\abs{u-k}-2k)/2$. Applying this identity and
\eqref{eq:sm-schmidtidentity} yields
\begin{equation}
 \int_{\sphere}\abs{r_x(\lambda)}\,\dd S
 =\pi\left(\norm{t_x+ce_z}+\norm{t_x-ce_z}-2c\right)
 =2\pi(1-c).
 \label{eq:sm-normalization}
\end{equation}
Replacing $c\abs{\lambda_z}$ by $c\lambda_z$ within the sum of the two
absolute values only exchanges its terms. The area density
$\abs{r_x}/[2\pi(1-c)]$ is therefore normalized. Relative to
$\mu=\dd S/(4\pi)$, it is precisely~\eqref{eq:qubitdensity}.

The density $f_x$ is even in $\lambda$, and
\begin{equation}
 a_x(\lambda)f_x(\lambda)=\frac{2r_x(\lambda)}{1-c}.
 \label{eq:sm-signeddensity}
\end{equation}
Both sign responses reverse under $\lambda\mapsto-\lambda$ except on great
circles. Since $T$ is invertible in the stated parameter range, these circles
have zero area and zero $f_x\mu$ probability. The two sign responses thus
have zero mean under $\nu_x=f_x\mu$. Their correlation defines the unbiased
behavior $P_N$ in~\eqref{eq:nonlocalbehavior}.

Combining~\eqref{eq:sm-localcorrelation} and~\eqref{eq:sm-signeddensity},
\begin{align}
 cC_L(x,y)+(1-c)C_N(x,y)
 &=\frac1{2\pi}\int_{\sphere}(q_x+r_x)b_y\,\dd S\nonumber\\
 &=\frac1{2\pi}t_x\cdot
       \int_{\sphere}\lambda b_y(\lambda)\,\dd S
 =x^{\mathsf T}Ty.
 \label{eq:sm-fullcorrelation}
\end{align}
The mixture has means $cx_z,cy_z$. For arbitrary signs $A,B$,
\begin{equation}
 \Prob(A=\alpha,B=\beta)
 =\frac14\left[1+\alpha\E A+\beta\E B+\alpha\beta\E(AB)\right],
 \label{eq:sm-binarymoments}
\end{equation}
as follows by multiplying
$\ind\{A=\alpha\}=(1+\alpha A)/2$ and the corresponding expression for $B$.
The three moments therefore prove
$P_\theta=cP_L+(1-c)P_N$ with the full Born probabilities, including their
biased marginals.

The weight $c$ is the maximum local mixture weight by
Ref.~\cite{Portmann2012}. Its value is an input to the present work, not a new
optimality assertion. Retaining this decomposition maximizes the probability
of an input-independent local branch, but does not minimize the PR count
within the remaining branch or the mean PR cost over all protocols.

\subsection{Uniform Lipschitz estimate}

The absolute-value and positive-part functions are 1-Lipschitz. Hence
\begin{align}
 \abs{f_x(\lambda)-f_x(\lambda')}
 &\le\frac2{1-c}\left(
       \abs{(Tx)\cdot(\lambda-\lambda')}
       +c\abs{\lambda_z-\lambda'_z}\right)\nonumber\\
 &\le\frac{2(1+c)}{1-c}\norm{\lambda-\lambda'}
 =2\cot^2\theta\norm{\lambda-\lambda'}.
 \label{eq:sm-Lipschitz}
\end{align}
We used $\norm{Tx}\le1$ and
$1+c=2\cos^2\theta$, $1-c=2\sin^2\theta$. Joint continuity follows directly
from~\eqref{eq:qubitdensity}. Normalization, continuity, and full support of
$\mu$ verify the hypotheses of Theorem~\ref{thm:selection}. Evenness permits
paired antipodal sampling cells, as used in Sec.~\ref{sm:grid}.

\section{Exact selection from a fixed finite list}
\label{sm:selection}

We state the general sampling result underlying~\eqref{eq:samplinggoal},
then prove it using a finite-mixture identity. Here $X$ and $Z$ need not be
spheres. Full support means that every nonempty open subset of $Z$ has
positive $\mu$-mass.

\begin{smtheorem}[Uniform finite selection]
\label{thm:selection}
Let $X,Z$ be nonempty compact metric spaces, $\mu$ a full-support Borel
probability on $Z$, and $f:X\times Z\to[0,\infty)$ continuous with
$\int_Z f_x\,\dd\mu=1$ for every $x$. There exists a finite positive-mass
Borel partition $C_1,\ldots,C_m$ with representatives $v_j\in C_j$ such that
\begin{equation}
 |f_x(z)-f_x(v_j)|\le\tfrac12,
 \qquad z\in C_j,\quad x\in X.
 \label{eq:sm-cellvariation}
\end{equation}
For any such partition, define $\gamma_j,\underline f_x,h_x,w_j,W_x$ by
\eqref{eq:lowerresidual}--\eqref{eq:cellweights}, with $z\in Z$ in place of
$\lambda\in\sphere$. Then $0\le h_x\le1$ and $W_x>0$. Independently generate
$\Lambda_j\sim\mu(\cdot\mid C_j)$ for $j\le m$ and
$\Lambda_N\sim\mu$, where $N=m+1$. The probabilities~\eqref{eq:selector}
define jointly Borel kernels satisfying
\begin{equation}
 \Omega\sim Q,\quad I\sim p(\cdot\mid x,\Omega)
 \quad\Longrightarrow\quad \Lambda_I\sim f_x\mu
 \quad\text{for every }x,
 \label{eq:selectiontheorem}
\end{equation}
where $Q$ is the joint distribution of the independent candidates, given
by $Q=[\bigotimes_{j=1}^m\mu(\cdot\mid C_j)]\otimes\mu$.
Given the partition, its
representatives and masses, the selector uses $m+1$ density evaluations,
one random decision about whether to use the global candidate, and at most
one random choice among the regions.
\end{smtheorem}

\subsection{Finite-mixture selection with a bounded residual}

The elementary identity below adapts mixture and acceptance-complement
sampling~\cite{Devroye1986} to a list generated before the target input is known.
A density is always relative to the specified reference probability $\mu$.

\begin{smlemma}[Finite-envelope selection]
\label{lem:envelope}
Let $X,Z$ be standard Borel spaces and $\mu$ a probability on $Z$. Suppose
$f:X\times Z\to[0,\infty)$ is Borel and $\int f_x\,\dd\mu=1$ for every
$x$. Let $g_1,\ldots,g_m$ be fixed normalized Borel densities, $w_j:X\to[0,\infty)$
Borel weights, and $h:X\times Z\to[0,1]$ Borel. Assume that, for each $x$,
\begin{equation}
 f_x(z)=\sum_{j=1}^m w_j(x)g_j(z)+h_x(z)
 \qquad\text{for }\mu\text{-almost every }z.
 \label{eq:sm-envelope}
\end{equation}
Independently draw $\Lambda_j\sim g_j\mu$ for $j\le m$ and
$\Lambda_N\sim\mu$, where $N=m+1$. There is a jointly Borel randomized
selection of an index $I$ such that $\Lambda_I\sim f_x\mu$ for every $x$.
\end{smlemma}

\begin{proof}
Normalization gives
\begin{equation}
 W_x:=\sum_{j=1}^m w_j(x)=1-\int h_x\,\dd\mu.
 \label{eq:sm-W}
\end{equation}
When $W_x>0$, accept the global candidate $\Lambda_N$ with probability
$h_x(\Lambda_N)$. If it is not accepted, choose $j$ with probability
$w_j(x)/W_x$ and return $\Lambda_j$. The global acceptance contributes the
subprobability $h_x\mu$. Its complement has probability $W_x$ and is
independent of all cell candidates. Thus, for every Borel $D\subseteq Z$,
\begin{align}
 \Prob(\Lambda_I\in D\mid x)
 &=\int_D h_x\,\dd\mu+
   W_x\sum_{j=1}^m\frac{w_j(x)}{W_x}\int_D g_j\,\dd\mu \nonumber\\
 &=\int_D f_x\,\dd\mu.
 \label{eq:sm-selectionidentity}
\end{align}
When $W_x=0$, all weights vanish and $h_x=1$ almost everywhere, so
$f_x=1$ almost everywhere. Define $I=N$ for all lists in this case.
The selector is then correct as well. The acceptance test and finite
categorical choice are Borel operations, and the case distinction uses the
Borel function $W_x$.
\end{proof}

The selection of a mixture component does not inspect that component's
candidate. Only the residual acceptance test inspects its candidate. This
independence is the reason a single pre-sampled point from each component
suffices; the construction contains no repeated acceptance trials.

\subsection{Uniform partitions and measurable selection}
\label{sm:selectionproof}

\begin{proof}[Proof of Theorem~\ref{thm:selection}]
Continuity of $f$ on the compact product $X\times Z$ gives a radius
$\delta>0$ such that
\begin{equation}
 d_Z(z,z')\le\delta
 \quad\Longrightarrow\quad
 \abs{f_x(z)-f_x(z')}\le\frac12
 \qquad\text{for every }x\in X.
 \label{eq:sm-uniformcontinuity}
\end{equation}
Choose a maximal set of centers $v_1,\ldots,v_m$ whose pairwise distances
are strictly greater than $\delta$. Compactness bounds its cardinality;
maximality gives a cover by closed radius-$\delta$ balls. Assign each point to its nearest center,
breaking ties by the least index. The resulting Voronoi cells are Borel,
contain their representatives, and satisfy~\eqref{eq:sm-cellvariation}.
Each cell contains the open radius-$\delta/2$ ball around its center, so
full support gives $\mu(C_j)>0$. Under a uniform Lipschitz bound $L>0$, one
may take $\delta=1/(2L)$.

Define $\gamma_j,\underline f_x,h_x$ as in~\eqref{eq:lowerresidual}. For $z\in C_j$ with
$f_x(v_j)\ge1/2$,
\begin{equation}
 h_x(z)=f_x(z)-f_x(v_j)+\frac12\in[0,1].
 \label{eq:sm-residualbound}
\end{equation}
If $f_x(v_j)<1/2$, then $\gamma_j(x)=0$ and
$0\le h_x(z)=f_x(z)\le1$. Hence the lower density is nonnegative and lies
below $f_x$, leaving a residual bounded by one. The choices
\begin{equation}
 g_j(z)=\frac{\ind\{z\in C_j\}}{\mu(C_j)},\qquad
 w_j(x)=\mu(C_j)\gamma_j(x)
 \label{eq:sm-cellmixture}
\end{equation}
satisfy Lemma~\ref{lem:envelope}.

It remains to verify that the denominator in~\eqref{eq:selector} is strictly
positive for every $x$. If $W_x=0$, positivity of all cell masses forces
$\gamma_j(x)=0$ for every $j$, and therefore $f_x(v_j)\le1/2$. The variation
bound gives $f_x\le1$ everywhere. Since $f_x$ is continuous,
$\int f_x\,\dd\mu=1$, and $\mu$ has full support, this implies
$f_x\equiv1$: any point with $f_x<1$ would have a positive-mass neighborhood
on which $f_x$ is bounded away from one. But $f_x\equiv1$ contradicts
$f_x(v_j)\le1/2$. Thus $W_x>0$.

The functions $\gamma_j,w_j,W_x$ are continuous in $x$, the residual is jointly
Borel, and $W_x$ never vanishes. Consequently all probabilities
$p(i\mid x,\omega)$ in~\eqref{eq:selector} are jointly Borel. If
$U\sim\Unif[0,1)$ is independent of the list, the rule
\begin{equation}
 I(x,\omega,U)=\min\left\{i:
             U<\sum_{k=1}^{i}p(k\mid x,\omega)\right\}
 \label{eq:sm-inversecdf}
\end{equation}
is a Borel randomized selection with these probabilities. It implements
the selector in Lemma~\ref{lem:envelope}, so the selected vector has the
target distribution $f_x\mu$. The required density
evaluations are those at the $m$ representatives and the global candidate.
\end{proof}

More explicitly, for a realized list $\omega$,
\begin{align}
 \Prob(\Lambda_I\in D\mid x,\Omega=\omega)
 &=\sum_{i=1}^N p(i\mid x,\omega)\ind\{\omega_i\in D\},\nonumber\\
 \Prob(\Lambda_I\in D\mid x)
 &=\int Q(\dd\omega)\sum_{i=1}^N
          p(i\mid x,\omega)\ind\{\omega_i\in D\}
   =\int_D f_x\,\dd\mu.
 \label{eq:sm-conditionedlaw}
\end{align}
The first line gives the conditional distribution of the selected vector
given a fixed list. Averaging over the random list yields the target
distribution $f_x\mu$, as shown in the second line.

Once the partition and masses are fixed, the selector requires finitely many
density evaluations, arithmetic operations, comparisons, and random draws.
For the Voronoi construction a direct cell search and cumulative-weight
calculation require $O(m)$ distance evaluations and arithmetic operations.
All local coins can be incorporated into the shared seed. These are finite
procedures in an exact-real model, not bounds on the number of bits needed
to represent inputs or samples.

\subsection{Relation to random-selection criteria}

Let $Q$ be the joint distribution of the candidate list. For a Borel set
$D\subseteq Z$, define
\begin{equation}
 H_D=\{\omega:\omega_i\in D\text{ for at least one }i\}.
 \label{eq:sm-hitset}
\end{equation}
Any selection with target distribution $\nu_x$ must satisfy
$\nu_x(D)\le Q(H_D)$. On compact metric spaces, the Strassen--Artstein
criterion makes this condition sufficient for a randomized
selection for each fixed target distribution~\cite{Strassen1965,Artstein1983}. The explicit probabilities above
construct such a selection directly; no separate measurable choice of
existentially specified selectors is required.

The distinction between independent and identically distributed candidates
is important. If $\nu_x(D)=1$ with $0<\mu(D)<1$, then a finite independent
list drawn entirely from $\mu$ misses $D$ with positive probability
$(1-\mu(D))^N$. It cannot supply an exact selection. In our construction,
any cell with $w_j(x)>0$ satisfies $f_x\ge \gamma_j(x)>0$ on that cell, while
the global candidate is accepted according to the residual. The construction
therefore accommodates target densities that vanish on nonempty regions.

\section{Explicit spherical partition and candidate generation}
\label{sm:grid}

\subsection{Paired angular cells}

Fix $0<\theta<\pi/4$, put $L=L_\theta=2\cot^2\theta\ge2$, and set
\begin{equation}
 n=\left\lceil\frac{\pi L}{\sqrt2}\right\rceil,
 \qquad \Delta=\frac{\pi}{2n},\qquad m=4n^2.
 \label{eq:sm-gridparameters}
\end{equation}
In polar coordinates $(\vartheta,\varphi)$ on the northern hemisphere,
define the closed patches
\begin{equation}
 \begin{split}
 D_{k\ell}=\bigl\{&(\sin\vartheta\cos\varphi,\sin\vartheta\sin\varphi,
                    \cos\vartheta):\\
 &k\Delta\le\vartheta\le(k+1)\Delta,\quad
 \ell\Delta\le\varphi\le(\ell+1)\Delta\bigr\},
 \end{split}
 \label{eq:gridpatches}
\end{equation}
where $0\le k<n$ and $0\le\ell<4n$. Pair each patch with its antipodal image.
The sets cover the sphere. Assign a point to the least index
$j=4nk+\ell+1$ among paired patches containing it. The resulting cells are
Borel and antipodally invariant. The boundary convention resolves poles,
angular edges, and the equator, all of which have zero area. Every cell has
positive area and contains its northern midpoint representative $v_{k\ell}$,
whose polar angles are $((k+\tfrac12)\Delta,(\ell+\tfrac12)\Delta)$.

For a point in a northern patch, the straight line in angular coordinates
from the representative to the point has increments of magnitude at most
$\Delta/2$ in each coordinate. The spherical line element is
$\dd\vartheta^2+\sin^2\vartheta\,\dd\varphi^2$. Its path length is bounded
by the Euclidean length in angular coordinates, and chordal distance is no
larger than path length. Consequently
\begin{equation}
 \min\{\norm{\lambda-v_{k\ell}},\norm{\lambda+v_{k\ell}}\}
 \le\frac{\Delta}{\sqrt2}
 =\frac{\pi}{2\sqrt2\,n}
 \le\frac1{2L}
 \qquad (\lambda\in C_{k\ell}).
 \label{eq:sm-gridradius}
\end{equation}
The estimate remains valid for the assigned boundary points. Since $f_x$ is
even,~\eqref{eq:sm-Lipschitz} implies
\begin{equation}
 \abs{f_x(\lambda)-f_x(v_{k\ell})}\le\frac12
 \qquad\text{for every }x\in\sphere,\quad\lambda\in C_{k\ell}.
 \label{eq:sm-gridvariation}
\end{equation}
Thus the paired cells satisfy the condition of Theorem~\ref{thm:selection}.

\subsection{Cell probabilities and exact candidate generation}

Integrating the spherical area element
$\sin\vartheta\,\dd\vartheta\,\dd\varphi$ over a northern patch, doubling
for the antipodal patch, and dividing by $4\pi$ gives
\begin{equation}
 \mu(C_{k\ell})
 =\frac{2\Delta}{4\pi}
     [\cos(k\Delta)-\cos((k+1)\Delta)]
 =\frac{\cos(k\Delta)-\cos((k+1)\Delta)}{4n}.
 \label{eq:sm-gridmass}
\end{equation}
These masses are positive and sum to one.

To generate the candidate for $C_{k\ell}$, independently draw
$t\sim\Unif[\cos((k+1)\Delta),\cos(k\Delta)]$,
$\varphi\sim\Unif[\ell\Delta,(\ell+1)\Delta)$, and a fair sign
$\varepsilon\in\{-1,+1\}$, and set
\begin{equation}
 \Lambda_{4nk+\ell+1}
 =\varepsilon\bigl(\sqrt{1-t^2}\cos\varphi,
                   \sqrt{1-t^2}\sin\varphi,t\bigr).
 \label{eq:gridcandidate}
\end{equation}
The global candidate uses the unsigned formula with independent
$t\sim\Unif[-1,1]$ and $\varphi\sim\Unif[0,2\pi)$. Draws for different
candidates are independent.

With $t=\cos\vartheta$, uniform surface area is proportional to
$\dd t\,\dd\varphi$. Independent uniform sampling of $t$ and $\varphi$
on the intervals specified before~\eqref{eq:gridcandidate} therefore produces
the normalized uniform distribution on the northern patch. An independent
fair sign chooses between that patch and its antipodal image. Boundaries
have zero area, so sampling from each assigned cell produces exactly
$\mu(\cdot\mid C_{k\ell})$. The unsigned formula with uniform
$t\in[-1,1]$ and $\varphi\in[0,2\pi)$ samples the global candidate uniformly
on the sphere, with distribution $\mu$.

This specifies $Q$ entirely through independent uniform real variables and
fair bits. Given $x$, Alice evaluates $f_x$ at the $m$ midpoint representatives
and at the global candidate. She locates the latter's paired cell, computes
$h_x$ and the weights in~\eqref{eq:cellweights}, and performs the finite
selection~\eqref{eq:selector}. The masses in Eq.~\eqref{eq:sm-gridmass},
pointwise evaluations of $f_x$, and finite sums determine all selection
probabilities. In particular, Alice need not evaluate integrals of $f_x$
over the cells. Exact evaluation of the formulas used is still an assumption
of the local model.

Together, Sec.~\ref{sm:quantum} and the grid variation bound verify all
hypotheses of Theorem~\ref{thm:selection} for the prescribed qubit list.
Consequently the vector selected by the main-text protocol has distribution
$\nu_x$ for every $x$, as in~\eqref{eq:samplinggoal}, while the joint
distribution $Q$ of the candidate list is fixed before the inputs arrive.

\section{Finite-box implementation and resource complexity}
\label{sm:index}

\subsection{Index-box realization and output statistics}
\label{sm:boxcorrectness}

For every $i,z$, exactly two pairs $u,v$ satisfy $u\oplus v=z_i$, and
\eqref{eq:indexbox} assigns probability $1/2$ to each. Both marginals are
therefore fair and independent of the other input.

For the implementation~\eqref{eq:prinputs}--\eqref{eq:proutputs}, the PR
relations give
\begin{equation}
 u\oplus v
 =z_N\oplus\bigoplus_{j<N}\ind\{I=j\}(z_j\oplus z_N)=z_I.
 \label{eq:sm-parity}
\end{equation}
The shared bit $H$ is independent of all inputs and preprocessing, so both
marginals are fair. This is exactly $R_N$. When $N=1$ the empty parity is
zero and the formula uses no PR box. No PR input depends on any PR output;
all calls are parallel.

Condition on fixed $x,y$, a realized list $\omega$, and Alice's chosen index
$I$. The bits $d,z$ are then fixed. The final outputs in~\eqref{eq:finaloutputs}
remain individually fair and satisfy
\begin{equation}
 AB=(-1)^{u\oplus v\oplus d}
   =(-1)^{z_I}(-1)^d
   =b_y(\Lambda_I)a_x(\Lambda_I).
 \label{eq:sm-outputproduct}
\end{equation}
Averaging over the list and the index and using Theorem~\ref{thm:selection}
gives
\begin{equation}
 \E(A\mid x,y)=\E(B\mid x,y)=0,\qquad
 \E(AB\mid x,y)=\int a_xb_y\,\dd\nu_x=C_N(x,y).
 \label{eq:sm-nonlocalmoments}
\end{equation}
Equation~\eqref{eq:sm-binarymoments} proves that this branch generates
$P_N$, not merely its correlation function. The shared branch coin and the
decomposition verified in Sec.~\ref{sm:quantum} then recover $P_\theta$.
This establishes exactness of the complete output distribution. We next
verify the resource bound and collect the remaining requirements of
Theorem~\ref{thm:main}.

\subsection{Worst-case and mean resource budgets}

The list has $N=m+1=4n^2+1$ candidates, and the index-box realization uses
$K=N-1=4n^2$ PR boxes. Substitution of~\eqref{eq:sm-gridparameters} gives
\begin{equation}
 K(\theta)=4\left\lceil\sqrt2\pi\cot^2\theta\right\rceil^2.
 \label{eq:sm-exactgridbudget}
\end{equation}
To recover the simpler upper bound, use $L\ge2$ and
\begin{equation}
 n\le\frac{\pi L}{\sqrt2}+1
 \le\left(\frac{\pi}{\sqrt2}+\frac12\right)L
 <2\sqrt2\,L.
 \label{eq:sm-ceilingbound}
\end{equation}
It follows that
\begin{equation}
 K(\theta)\le32L^2=128\cot^4\theta\le128\theta^{-4},
 \label{eq:sm-coarsebound}
\end{equation}
where the last inequality uses $\tan\theta\ge\theta$.
The explicit budget also has the expansion
$K(\theta)=8\pi^2\cot^4\theta+O(\cot^2\theta)$ as $\theta\to0^+$.
These statements concern the displayed construction, not the minimal PR cost.

The PR branch occurs with probability $1-c=2\sin^2\theta$. If calls are
omitted on the local branch,
\begin{equation}
 \E K_{\mathrm{used}}=(1-c)K(\theta)
 \le256\frac{\cos^4\theta}{\sin^2\theta}
 =O(\theta^{-2}).
 \label{eq:sm-meanbudget}
\end{equation}
The branch probability is input-independent. This average estimate is separate
from the finite worst-case bound on every run.

\subsection{Completion of the main theorem}
\label{sm:mainproof}

\begin{proof}[Proof of Theorem~\ref{thm:main}]
Section~\ref{sm:quantum} proves the decomposition~\eqref{eq:decomposition}
and the properties of its target densities. Theorem~\ref{thm:selection},
applied to the explicit partition in Sec.~\ref{sm:grid}, proves exact
selection from the input-independent finite list. The index-box construction
and the moment calculation in Sec.~\ref{sm:boxcorrectness} then give the
Born distribution for every pair $x,y$.

Shared randomness is prepared using distributions that depend only on the
fixed state. The branch coin,
list, bit $H$, and local coins can be included in a single seed independent
of the measurement inputs, with the PR boxes independent of that seed.
The selection kernel is jointly Borel, and all remaining local operations
are Borel sign tests, arithmetic, and finite stochastic processing. Each
branch returns both outputs on every run; no trial is discarded. Each
party uses only its own input and local PR outputs, so the construction
requires no communication. Its calls are parallel and number at most
$K(\theta)$ by~\eqref{eq:sm-exactgridbudget}--\eqref{eq:sm-coarsebound}.
The local branch occurs with probability $c$ and uses no box. Its maximal
weight is the established result of Ref.~\cite{Portmann2012}.
\end{proof}

\subsection{Exact PR-box cost of the full index box}
\label{sm:indexoptimality}

The following statement concerns the complete domain
$\{1,\ldots,N\}\times\{0,1\}^N$, not just input pairs arising from a
particular qubit preprocessing rule. We use the standard wiring model of
independent, singly used PR boxes: each party may choose its inputs and local
query order from its own input, randomness, and previous local outputs,
without communication~\cite{Kaplan2011}.

\begin{smprop}[Index-box cost]
\label{prop:indexcost}
For every integer $N\ge1$, exact realization of $R_N$ on its full input domain
requires at least $N-1$ PR-box calls in the worst case, even with arbitrary
shared randomness and local adaptive wirings. The parallel realization
\eqref{eq:prinputs}--\eqref{eq:proutputs} attains this bound.
\end{smprop}

\begin{proof}
The upper bound was proved above. For the lower bound, suppose a protocol
uses at most $k$ PR boxes; pad unused calls with ignored dummy calls. Only the
zero-error parity requirement $u\oplus v=z_i$ is needed.

We apply the one-way communication reduction for nonlocal-box
protocols~\cite[Lemma~1]{Kaplan2011}, with the communication directed from
Bob to Alice. Because the input domain is finite, all shared and local
random seeds can be fixed so that the parity condition holds with probability
one over the PR outputs simultaneously for every input pair. In the resulting
protocol consider the transcript in which all of Bob's $k$ PR outputs are
zero. Its probability is $2^{-k}$ for every input pair: each fresh local
PR output is fair conditional on Bob's previous local transcript, regardless
of Alice's inputs or query order. Bob can therefore
compute his entire vector of PR input bits, indexed by the fixed box labels,
and his final output $v$, from $z$ alone. This remains true when his local
query order is adaptive.

Bob sends these $k$ input bits and the final bit $v$ to Alice. Alice simulates
her own query order. For a box with her input $a$ and Bob's transmitted input
$b$, her output on the selected transcript must be $ab$, since Bob's box
output is zero. She thus reconstructs her final $u$ and obtains
$z_i=u\oplus v$. The reduction produces a deterministic one-way protocol
for indexed retrieval, with at most $k+1$ transmitted bits.

Distinct strings $z,z'\in\{0,1\}^N$ must give different messages. Otherwise,
choosing an index $i$ where they differ would force Alice to give the same
answer on two inputs requiring different answers. There must therefore be
at least $2^N$ messages, so $k+1\ge N$. This proves $k\ge N-1$.
\end{proof}

This lower bound is an application of the communication method, not a lower
bound for all quantum simulators. In the qubit protocol the table has the
restricted form $z_j=\ind\{y\cdot\Lambda_j<0\}$, and alternative protocols
need not implement the full index box or use the same candidate list.

\subsection{Uniform resources away from the product endpoint}
\label{sm:uniform}

Fix $0<\theta_0<\pi/4$. For every $\theta\in[\theta_0,\pi/4)$,
$L_\theta\le L_{\theta_0}$. Hence the grid, candidate distributions, and
index box constructed at $\theta_0$ serve all these states. Only the target
density evaluations, the selection weights, and the branch probability
change with the known state. The maximally entangled endpoint may use the
same PR supply by embedding its one-box protocol and ignoring unused boxes.
At the product endpoint no box is needed. The divergence of our bound near
that endpoint does not establish a divergence of the minimum simulation cost.

\section{Mixed states and binary measurements}
\label{sm:extensions}

\subsection{Endpoints, basis changes, and fixed mixed states}

At $\theta=0$, the state is product and its measurement probabilities factor;
local sampling suffices. At $\theta=\pi/4$, the one-PR-box simulation of
Ref.~\cite{Cerf2005} applies. Every pure two-qubit state is related to a
Schmidt state by local unitaries. Conjugating each local measurement by its
known unitary absorbs these basis changes without communication or additional
boxes. Thus all fixed pure two-qubit states are covered.

For a fixed mixed state, choose a finite decomposition
\begin{equation}
 \rho=\sum_{r=1}^{q}p_r\ket{\psi_r}\bra{\psi_r},
 \qquad p_r\ge0,\qquad\sum_rp_r=1.
 \label{eq:sm-mixture}
\end{equation}
A spectral decomposition has $q\le4$. Shared randomness selects $r$ with
probability $p_r$, independently of the measurements. The parties run the
pure-state simulator for $\ket{\psi_r}$. Writing $\theta_r$ for its Schmidt
parameter and $K_r=K(\theta_r)$ for its sufficient PR budget,
a supply of
\begin{equation}
 K(\rho)=\max_{r:p_r>0}K_r
 \label{eq:sm-mixedbudget}
\end{equation}
covers every branch, with unused boxes ignored. Linearity of Born's rule
proves the mixed-state statistics. Finite mixing preserves the jointly Borel
local kernels and parallel PR use. The chosen decomposition determines an
upper bound and is not asserted to optimize the cost.

\subsection{Binary POVMs and local quantum channels}
\label{sm:binarypovm}

Alice's input is a binary qubit POVM $\{E,\id_2-E\}$, with
$0\le E\le\id_2$. She obtains a direction $x$ and eigenvalues $e_\pm$
from the spectral decomposition of its effect $E$:
\begin{equation}
 E=e_+\Pi_{+|x}+e_-\Pi_{-|x},\qquad
 \Pi_{\pm|x}=\frac{\id_2\pm x\cdot\sigma}{2},\qquad
 0\le e_\pm\le1.
 \label{eq:sm-effect}
\end{equation}
Alice supplies this derived direction $x$ to the projective-measurement
simulator. Conditional on its
outcome $\alpha\in\{-1,+1\}$, she reports the first POVM outcome with probability
$e_\alpha$ and the second otherwise. Bob uses the analogous rule for
$\{F,\id_2-F\}$. This is local stochastic postprocessing and adds no PR
calls~\cite{Oszmaniec2017}.

If $q_A(j\mid\alpha)$ and $q_B(k\mid\beta)$ are the local transition probabilities,
with the postprocessing coins conditionally independent given $\alpha,\beta$, then
\begin{align}
 P(j,k)
 &=\sum_{\alpha,\beta}q_A(j\mid\alpha)q_B(k\mid\beta)
       \Tr[\rho(\Pi_{\alpha|x}\otimes\Pi_{\beta|y})]\nonumber\\
 &=\Tr[\rho(E_j\otimes F_k)].
 \label{eq:sm-povmidentity}
\end{align}
For measurability, write $E=t\id_2+v\cdot\sigma$. At $v\ne0$, choose
$x=v/\norm v$ and $e_\pm=t\pm\norm v$; at $v=0$, use a fixed direction
and $e_+=e_-=t$. These are Borel choices. The same holds for Bob. Combining
this construction with the previous subsection proves
Corollary~\ref{cor:mixed}. General multi-outcome POVMs need not admit this
projective-measurement postprocessing representation~\cite{Oszmaniec2017}.

A local quantum channel $\mathcal E$ followed by a binary measurement is
also covered. Its Heisenberg adjoint $\mathcal E^*$ takes the effect $E$ to
$\mathcal E^*(E)$; positivity and unitality give another binary POVM
$\{\mathcal E^*(E),\id_2-\mathcal E^*(E)\}$. The channel may be included in
the local input. This statement concerns classical outcome statistics only.

\section{Output stability under sampling errors}
\label{sm:stability}

For probability measures $\sigma,\tau$ on the same measurable space, define
\begin{equation}
 \norm{\sigma-\tau}_{\tv}
 :=\sup_D\abs{\sigma(D)-\tau(D)},
 \label{eq:sm-TVdefinition}
\end{equation}
where $D$ ranges over measurable sets. Suppose a perturbed implementation
produces a selected vector with probability distribution $\widetilde\nu_x$
satisfying
\begin{equation}
 \sup_x\norm{\widetilde\nu_x-\nu_x}_{\tv}\le\epsilon.
 \label{eq:sm-samplingerror}
\end{equation}
Assume the index box and sign responses are evaluated exactly and remain
otherwise unchanged. Since $\abs{a_xb_y}=1$, the correlation changes by at
most $2\epsilon$. The fair index-box outputs preserve zero means regardless
of the distribution of the selected vector. For two unbiased binary behaviors, total-variation
distance is half the absolute difference of their correlations, so
\begin{equation}
 \sup_{x,y}\norm{\widetilde P_N(\cdot\mid x,y)-P_N(\cdot\mid x,y)}_{\tv}
 =\frac12\sup_{x,y}\abs{\widetilde C_N(x,y)-C_N(x,y)}
 \le\epsilon.
 \label{eq:sm-nonlocalstability}
\end{equation}
If the local branch and branch coin are exact, then
\begin{equation}
 \sup_{x,y}\norm{\widetilde P_\theta(\cdot\mid x,y)
                    -P_\theta(\cdot\mid x,y)}_{\tv}
 \le(1-c)\epsilon.
 \label{eq:sm-fullstability}
\end{equation}
Local stochastic postprocessing does not increase this distance, and finite
mixtures preserve the corresponding weighted error bounds. For a Bell-game
payoff valued in $[0,1]$ and a fixed input distribution, the change in expected
payoff is bounded by the same uniform output-distribution error.

These bounds are conditional on~\eqref{eq:sm-samplingerror}. They do not
establish an end-to-end finite-precision implementation. In particular, a
finite or countable quantization of an absolutely continuous distribution
has total-variation distance one from the latter: the quantized support has
probability one under one measure and zero under the other. Numerical
quantization and approximate sign evaluation therefore require their own
output-level error analysis. The exact theorem uses the prescribed
continuous variables and exact local evaluations.

\end{document}